\documentclass[submission,copyright,creativecommons]{eptcs}
\providecommand{\event}{NCMA 2026} 

\usepackage[english]{babel}
\usepackage{iftex}
\usepackage{hyperref}
\usepackage{tikz,pgf}
\usetikzlibrary{positioning, arrows, automata}
\usepackage{comment}
\usepackage{amsmath, amsfonts, amssymb, theorem, enumerate}

\newcommand{\NIL}{\mathit{NIL}}
\newcommand{\COMB}{\mathit{COMB}}
\newcommand{\DEF}{\mathit{DEF}}
\newcommand{\SYDEF}{\mathit{SYDEF}}
\newcommand{\COMM}{\mathit{COMM}}
\newcommand{\CIRC}{\mathit{CIRC}}
\newcommand{\SUF}{\mathit{SCL}}
\newcommand{\PRE}{\mathit{PCL}}
\newcommand{\INF}{\mathit{ICL}}
\newcommand{\NC}{\mathit{NC}}
\newcommand{\SF}{\mathit{SF}}
\newcommand{\UF}{\mathit{UF}}
\newcommand{\FIN}{\mathit{FIN}}
\newcommand{\MON}{\mathit{MON}}
\newcommand{\PS}{\mathit{PS}}
\newcommand{\ORD}{\mathit{ORD}}
\newcommand{\STAR}{\mathit{STAR}}
\newcommand{\COM}{\mathit{COM}}
\newcommand{\LCOM}{\mathit{LCOM}}
\newcommand{\RCOM}{\mathit{RCOM}}
\newcommand{\TCOM}{\mathit{2COM}}
\newcommand{\INFF}{\mathit{IFR}}
\newcommand{\SUFF}{\mathit{SFR}}
\newcommand{\PREF}{\mathit{PFR}}
\newcommand{\REG}{\mathit{REG}}
\newcommand{\Comm}{\mathit{Comm}}
\newcommand{\Circ}{\mathit{Circ}}
\newcommand{\Suf}{\mathit{Suf}}
\newcommand{\Inf}{\mathit{Inf}}
\newcommand{\Pre}{\mathit{Pre}}
\newcommand{\ec}[1]{\ensuremath{\mathcal{EC}(#1)}}

\newcommand{\ellA}{\ell_{\mathrm{A}}}
\newcommand{\ellC}{\ell_{\mathrm{C}}}
\newcommand{\qmand}{\quad\mbox{and}\quad}

\newcommand{\iirule}[4]{
\begin{align*}
  S_1 &= #1, & S_2 &= #3,\\
  C_1 &= #2, & C_2 &= #4
\end{align*}
}

\def\Set#1#2{\left\{\: #1\;|\; #2\:\right\}}

\def\set#1#2{\{\; #1 \mid #2\;\}}
\def\sets#1{\{#1\}}

\def\cF{{\cal F}}

\def\cS{{\cal S}}

\def\cEC{{\cal EC}}
\def\cIC{{\cal IC}}

\DeclareSymbolFont{symbols}{OMS}{cmsy}{m}{n}

\def\Lra{\Longrightarrow}

\newtheorem{theorem}{Theorem}
\newtheorem{lemma}[theorem]{Lemma}
\newtheorem{corollary}[theorem]{Corollary}
\theorembodyfont{\rm}
\newtheorem{example}[theorem]{Example}
\newenvironment{proof}{{\em Proof. }}{{}\hspace*{\fill}$\Box$ \par \medskip }
\newenvironment{proof*}{{\em Proof. }}{\par \medskip }

\newlength{\btlabelwidth}
\newlength{\btleftmargin}
\newenvironment{btlists}{\begin{list}{{\rm--}}{%
\setlength{\labelwidth}{\btlabelwidth}\setlength{\leftmargin}{\btleftmargin}%
\setlength{\topsep}{0pt plus0.2ex}%
\setlength{\itemsep}{0ex plus0.2ex}%
\setlength{\parsep}{0pt plus0.2ex}}}{\end{list}}
\tikzstyle{to}=[->, >=stealth]
\tikzstyle{hier}=[->, >=angle 60]
\tikzstyle{hiero}=[->, >=angle 60, dashed]
\tikzstyle{state}=[circle,draw,inner sep=2pt,minimum size=8mm]
\tikzstyle{edgeLabel}=[inner sep=0.5mm,fill=white,text=black]

\usepackage{pdflscape}
\ifpdf
  \usepackage{underscore}     
  \usepackage[T1]{fontenc}    
\else
  \usepackage{breakurl}       
\fi

\title{Idefix-Free Languages and Their Application in External Contextual Grammars}
\author{Marvin K\"odding
\institute{Institut f{\"u}r Mathematik und Informatik, P{\"a}dagogische Hochschule Heidelberg\\Im Neuenheimer Feld 561, 69120 Heidelberg, Germany}
\email{koedding@ph-heidelberg.de}
\and
Bianca Truthe
\institute{Institut f\"ur Informatik, Universit\"at Giessen\\Arndtstr. 2, 35392 Giessen, Germany}
\email{bianca.truthe@informatik.uni-giessen.de}
}
\def\titlerunning{Idefix-Free Languages and Their Application in External Contextual Grammars}
\def\authorrunning{M. K\"odding \& B. Truthe}

\begin{document}
\maketitle

\begin{abstract}
In this paper, we continue the research on the power of contextual grammars with selection languages from 
subfamilies of the family of regular languages. 
We investigate infix-, prefix-, and suffix-free languages (referred to as idefix-free languages) and compare
such language families to some other subregular families of languages (finite, monoidal, nilpotent, combinational, 
(symmetric) definite, ordered, non-counting,  power-sepa\-rating, commutative, circular, union-free, star, and comet 
languages). Further, we compare the families of the hierarchies obtained for external 
contextual grammars with the language families defined by these new types for the selection. In this way, 
we extend the existing hierarchies by new language families.
\end{abstract}


\section{Introduction}

Contextual grammars were introduced by Solomon Marcus in \cite{Marcus.1969} as a formal model for the generation of natural languages. The derivation steps consist in adding contexts to given well-formed sentences, starting from an initial finite basis. Formally, a context is given by a pair $(u, v)$ of words. The external adding to a word $x$ gives the word $uxv$; the internal adding to a word $x$ yields all words $x_1 u x_2 v x_3$ where $x_1x_2x_3 = x$. In order to control the derivation process, contextual grammars with selection in a certain family of languages were defined. In such contextual grammars, a context $(u, v)$ may be added only around a word $x$ if this word belongs to a language which is associated with the context.

The study of external contextual grammars with selection in special regular sets was started by J\"urgen Dassow in \cite{Dassow.2005} and continued by J\"urgen Dassow, Florin Manea, and Bianca Truthe (see \cite{Dassow_Manea_Truthe.2012}), where further subregular families of selection languages were considered. The internal derivation mode with subregular selection languages was investigated in \cite{DasManTru12b}.

In the present paper, we continue this line of research from \cite{Koedding.Truthe.NCMA24.JALC.2025} and \cite{Koedding.Truthe.2025.idefix}. We extend the hierarchy of subregular language families by families of infix-free, prefix-free, and suffix-free languages (referred to collectively as \emph{idefix-free} languages). 

The paper is organized as follows. In Section 2, we recall some basic notions and define the subregular language families as well as contextual grammars. In Section 3, we investigate the relations between the families of idefix-free languages and other known subregular language families. In Section 4, we consider external contextual grammars with idefix-free selection languages and compare their generative capacity to that of grammars with other subregular selection languages.

\section{Preliminaries}

Throughout the paper, we assume that the reader is familiar with the basic concepts of the theory of automata 
and formal languages. For details, we refer to \cite{Rozenberg_Salomaa.1997}. Here we only recall some notations, 
definitions, and previous results which we need for the present research.

An alphabet is a non-empty finite set of symbols. For an alphabet $V$, we denote by~$V^*$ and $V^+$ the set of all 
words and the set of all non-empty words over $V$, respectively. The empty word is denoted by~$\lambda$. 
For a word $w$ and a letter $a$, we denote the length of $w$ by $|w|$ and the number
of occurrences of the letter~$a$ in the word $w$ by $|w|_a$. For a set $A$, we denote its cardinality by $|A|$.

The reversal of a word $w$ is denoted by $w^R$: if $w=x_1x_2\ldots x_n$ for letters $x_1,\ldots,x_n$, 
then $w^R=x_nx_{n-1}\ldots x_1$. By~$L^R$, we denote the language of all reversals of the words in $L$: 
$L^R=\set{w^R}{w\in L}$.

A deterministic finite automaton is a quintuple 
\[{\cal A}=(V,Z,z_0,F,\delta)\]
where $V$ is a finite set of input symbols, $Z$
is a finite set of states, $z_0\in Z$ is the initial state, $F\subseteq Z$ is a set of accepting states, and $\delta$ is
a transition function $\delta: Z\times V\to Z$. The language accepted by such an automaton is the set of all input words 
over the alphabet $V$ which lead letterwise by the transition function from the initial state to an accepting state.

All the languages accepted by a finite automaton
are called regular and form a
family denoted by~$\REG$.
Any subfamily of this set is called a subregular language family.

For a language $L$ over an alphabet $V$, we set
\[\Comm(L) = \set{a_{i_1}\dots a_{i_n}}{a_1\dots a_n\in L, \ n\geq 1,\ \{i_1,i_2,\dots ,i_n\}= \{1,2,\dots ,n \}}\]
as the commutative closure (the set of all permutations of words) of the language~$L$ and
\[\Circ(L) = \set{ vu }{ uv\in L,\ u,v\in V^* }\]
as the circular closure (the set of all circular shifts of words) of the language $L$.

For a language $L$ over an alphabet $V$, we set
\begin{align*}
\Inf(L) &=\set{y}{xyz\in L \text{ for some } x,z\in V^*},\\
\Pre(L) &=\set{x}{xy\in L \text{ for some } y\in V^*},\\
\Suf(L) &=\set{y}{xy\in L \text{ for some } x\in V^*}
\end{align*}
as the infix-, prefix-, and suffix-closure of $L$, respectively.
If the language $L$ is regular, then also $\Inf(L)$, $\Pre(L)$, and $\Suf(L)$ are regular.

\subsection{Some Subregular Language Families}

We consider the following restrictions for regular languages. In the following list of properties, we give already
the abbreviation which denotes the family of all languages with the respective property.
Let $L$ be a regular language over an alphabet $V$. With respect to the alphabet $V$, the language $L$ is said to be
\begin{btlists}
	\item \emph{monoidal} ($\MON$) if and only if $L=V^*$,
	\item \emph{nilpotent} ($\NIL$) if and only if it is finite or its complement $V^*\setminus L$ is finite,
	\item \emph{combinational} ($\COMB$) if and only if it has the form
	$L=V^*X$
	for some subset~$X\subseteq V$,
	\item \emph{definite} ($\DEF$) if and only if it can be represented in the form
	$L=A\cup V^*B$
	where~$A$ and~$B$ are finite subsets of $V^*$,
	\item \emph{symmetric definite} ($\SYDEF$) if and only if $L = EV^*H$ for some regular languages $E$ and $H$,
	\item \emph{prefix-free} ($\PREF$) if and only if no word in $L$ is a proper prefix of another word in $L$; formally, for any $x \in L$ and $y \in V^*$, the relation $xy \in L$ implies $y = \lambda$,
	\item \emph{suffix-free} ($\SUFF$) if and only if no word in $L$ is a proper suffix of another word in $L$; formally, for any $y \in L$ and $x \in V^*$, the relation $xy \in L$ implies $x = \lambda$,
	\item \emph{infix-free} ($\INFF$) if and only if no word in $L$ is a proper infix of another word in $L$; formally, for any $y \in L$ and $x, z \in V^*$, the relation $xyz \in L$ implies $x = z = \lambda$,
	\item \emph{infix-closed} ($\INF$) if
	and only if, for any three words over $V$, say $x\in V^*$, $y\in V^*$ and $z \in V^*$, the relation $xyz\in L$ implies
	the relation~$y\in L$ (equivalently, $L=\Inf(L)$),
	\item \emph{prefix-closed} ($\PRE$) if
	and only if, for any two words over $V$, say $x\in V^*$ and $y \in V^*$, the relation $xy\in L$ implies
	the relation~$x\in L$ (equivalently, $L=\Pre(L)$),
	\item \emph{suffix-closed} ($\SUF$) if
	and only if, for any two words over $V$, say $x\in V^*$ and~$y\in V^*$, the relation $xy\in L$ implies
	the relation~$y\in L$ (equivalently, $L=\Suf(L)$),
	\item \emph{ordered} ($\ORD$) if and only if the language is accepted by some deterministic finite
	automaton ${\cal A}=(V,Z,z_0,F,\delta)$
	with an input alphabet $V$, a finite set $Z$ of states, a start state $z_0\in Z$, a set $F\subseteq Z$ of
	accepting states and a transition mapping $\delta$ where $(Z,\preceq )$ is a totally ordered set and, for
	any input symbol~$a\in V$, the relation $z\preceq z'$ implies~$\delta (z,a)\preceq \delta (z',a)$,
	\item \emph{commutative} ($\COMM$) if and only if it contains with each word also all permutations of this
	word (equivalently, $L=\Comm(L)$),
	\item \emph{circular} ($\CIRC$) if and only if it contains with each word also all circular shifts of this
	word (equivalently, $L=\Circ(L)$),
	\item \emph{non-counting} ($\NC$) if and only if there is a natural
	number $k\geq 1$ such that, for any three words~$x\in V^*$, $y\in V^*$, and $z\in V^*$, it
	holds~$xy^kz\in L$ if and only if~$xy^{k+1}z\in L$,
	\item \emph{star-free} ($\SF$) if and only if $L$ can be described by a regular expression which is built by
	concatenation, union, and complementation,
	\item \emph{power-separating} ($\PS$) if and only if, there is a natural number $m\geq 1$ such that
	for any word~$x\in V^*$, either $J_x^m \cap L = \emptyset$ or $J_x^m\subseteq L$
	where $J_x^m = \set{ x^n}{n\geq m}$,
	\item \emph{union-free} ($\UF$) if and only if $L$ can be described by a regular expression which
	is only built by concatenation and Kleene closure,
	\item \emph{star} ($\STAR$) if and only if $L = H^*$ for some regular language $H \subseteq V^*$,
	\item \emph{left-sided comet} ($\LCOM$) if and only if $L = EG^*$ for some regular language $E$ and a regular
	language $G \notin \{\emptyset, \{\lambda\}\}$,
	\item \emph{right-sided comet} ($\RCOM$) if and only if $L = G^*H$ for some regular language $H$ and a regular
	language $G \notin \{\emptyset, \{\lambda\}\}$,
	\item \emph{two-sided comet} ($\TCOM$) if and only if $L = EG^*H$ for two regular languages $E$ and $H$ and a
	regular language $G \notin \{\emptyset, \{\lambda\}\}$.
\end{btlists}

We remark that monoidal, nilpotent, combinational, (symmetric) definite, ordered, star-free, 
union-free, star, and (left-, right-, or two-sided) comet languages are regular, whereas non-regular languages 
of the other types mentioned above exist.
Here, we consider among the infix-closed, prefix-closed, suffix-closed, infix-free, prefix-free, suffix-free, 
commutative, circular, non-counting, 
and power-separating languages only those which are also regular.
By $\FIN$, we denote the family of languages with finitely many words.
In \cite{McNaughton_Papert.1971}, it was shown that the families of the regular non-counting languages and 
the star-free languages are equivalent~($\NC=\SF$).

Some properties of the languages of the classes mentioned above can be found in
\cite{Shyr.1991} (monoids),
\cite{Gecseg_Peak.1972} (nilpotent languages),
\cite{Havel.1969} (combinational and commutative languages),
\cite{Perles_Rabin_Shamir.1963} (definite languages),
\cite{Paz_Peleg.1965} (symmetric definite languages),
\cite{Brzozowski_Jiraskova_Zou.2014} (prefix-closed languages),
\cite{Gill_Kou.1974} and \cite{Brzozowski_Jiraskova_Zou.2014} (suffix-closed languages),
\cite{Shyr_Thierrin.1974.ord} (ordered languages),
\cite{Kudlek.2004} (circular languages),
\cite{McNaughton_Papert.1971} (non-counting and 
star free
languages),
\cite{Shyr_Thierrin.1974.ps} (power-separating languages),
\cite{Brzozowski.1962} (union-free languages),
\cite{Brzozowski.1967} (star languages),
\cite{Brzozowski_Cohen.1969} (comet languages).

\subsection{Contextual Grammars}

Let $\cF$ be a family of languages. A contextual grammar with selection in $\cF$ is a triple~$G=(V,\cS,A)$
where
\begin{btlists}
\item $V$ is an alphabet, 
\item $\cS$ is a finite set of selection pairs $(S,C)$ with a selection language $S$ over some subset $U$ of the
alphabet $V$ which belongs to the family $\cF$ with respect to the alphabet $U$ and a finite
set~\hbox{$C\subset V^*\times V^*$} of contexts where, for each context~$(u,v)\in C$, at least one side is not 
empty: $uv\not=\lambda$,
\item $A$ is a finite subset of $V^*$ (its elements are called axioms).
\end{btlists}
We write a selection pair $(S,C)$ also as $S\to C$. In the case that $C$ is a singleton set~$C=\{(u,v)\}$, we
also write $S\to(u,v)$.
For a contextual grammar 
$G=(V,\sets{S_1\to C_1, S_2\to C_2,\dots, S_n\to C_n},A)$,
we set
\[\ellA(G) = \max \Set{ |w| }{ w\in A },\
\ellC(G) = \max \Set{ |uv| }{(u,v)\in C_i, 1\leq i\leq n},\ 
\ell(G)   = \ellA(G)+\ellC(G)+1.\]
We now define the derivation modes for contextual grammars with selection.

Let $G=(V,\cS,A)$ be a contextual grammar with selection.
A direct external derivation step in $G$ is defined as follows: a word~$x$ derives a word $y$ 
(written as~$x\Lra_\mathrm{ex} y$) if and only if there is a pair~$(S,C)\in\cS$ such that~$x\in S$ and $y=uxv$ 
for some pair $(u,v)\in C$.
Intuitively, one can only wrap a context $(u,v)\in C$ around a word $x$ if $x$ belongs to the corresponding
selection language $S$.

A direct internal derivation step in $G$ is defined as follows: a word~$x$ derives a word~$y$ 
(written as~$x\Lra_\mathrm{in} y$) if and only if there are words $x_1$,~$x_2$,~$x_3$
with~$x_1x_2x_3=x$ and there is a selection pair~$(S,C)\in\cS$ such that $x_2\in S$ and~$y=x_1ux_2vx_3$ 
for some pair $(u,v)\in C$.
Intuitively, we can only wrap a context~$(u,v)\in C$ around a subword~$x_2$ of~$x$ if $x_2$ belongs to 
the corresponding selection language~$S$.
%

By $\Lra^*_\mu$ we denote the reflexive and transitive closure of the relation~$\Lra_\mu$ 
for~$\mu\in\sets{\mathrm{ex},\mathrm{in}}$. 
The language generated by $G$ is defined as
$L_\mu(G)=\set{ z }{ x\Lra^*_\mu z \mbox{ for some } x\in A }.$
We omit the index $\mu$ if the derivation mode is clear from the context.

By~$\cEC(\cF)$, we denote the family of all languages generated externally by contextual grammars 
with selection in $\cF$. When a contextual grammar works in the external mode, we call it an external 
contextual grammar. 

\section{Results on families of idefix-free languages
}

In this section, we investigate inclusion relations between various subregular language classes.
Figure~\ref{fig:lang_erg_1} shows the results. 

\begin{figure}[htb]
	\centering
	\scalebox{.9}{\begin{tikzpicture}[node distance=15mm and 16mm, on grid]
			\node (MON) {$\MON$};
			\node (d0)[above=of MON] {};
			\node (FIN)[right=of d0] {$\FIN$};
			\node (NIL)[above=of d0] {$\NIL$};
			\node (COMB)[left=of NIL] {$\COMB$};
			\node (DEF)[above=of NIL] {$\DEF$};
			\node (d1)[left=of DEF] {};
			\node (SYDEF)[left=of d1] {$\SYDEF$};
			\node (d2)[right=of NIL] {};
			\node (ORD)[above=of DEF] {$\ORD$};
			\node (INF)[right=of d2] {$\INF$};
			\node (PRE)[right=of ORD] {$\PRE$};
			\node (SUF)[right=of PRE] {$\SUF$};
			\node (NC)[above=of ORD] {$\NC\stackrel{\text{\cite{McNaughton_Papert.1971}}}{=}\SF$};
			\node (PS)[above=of NC] {$\PS$};
			\node (RCOM)[above=of SYDEF] {$\RCOM$};
			\node (LCOM)[left=of RCOM] {$\LCOM$};
			\node (TCOM)[above=of RCOM] {$\TCOM$};
			\node (INFF)[right=of INF] {\small$\INFF$};
			\node (PREF)[right=of SUF] {\small$\PREF$};
			\node (SUFF)[right=of PREF] {\small$\SUFF$};
			\node (CIRC)[right=of SUFF] {$\CIRC$};
			\node (COMM)[below=of CIRC] {$\COMM$};
			\node (UF)[left=of LCOM] {$\UF$};
			\node (STAR)[below=of UF] {$\STAR$};
			\node (REG) [above of = PS] {$\REG$};

			\draw[hier, bend left] (MON) to node[edgeLabel] {\small\cite{Koedding.Truthe.2024}} (STAR);
			\draw[hier, bend left] (MON) to node[edgeLabel] {\small\cite{Koedding.Truthe.2024}} (SYDEF);
			\draw[hier] (RCOM) to node[pos=.45,edgeLabel]{\small\cite{Bordihn_Holzer_Kutrib.2009}}(TCOM);
			\draw[hier] (LCOM) to node[edgeLabel]{\small\cite{Koedding.Truthe.2024}}(TCOM);
			\draw[hier, bend left=20] (TCOM) to node[edgeLabel]{\small\cite{Bordihn_Holzer_Kutrib.2009}}(REG);
			\draw[hier] (STAR) to node[pos=.45,edgeLabel]{\small\cite{Koedding.Truthe.2024}} (UF);
			\draw[hier, bend left] (UF) to node[edgeLabel]{\small\cite{Holzer_Truthe.2015}}(REG);
			\draw[hier, bend right=40] (MON) to node[pos=.7,edgeLabel]{\small\cite{Truthe.2018}} (COMM);
			\draw[hier] (COMM) to node[pos=.45,edgeLabel]{\small\cite{Holzer_Truthe.2015}}(CIRC);
			\draw[hier, bend right=30] (CIRC) to node[edgeLabel]{\small{\cite{Holzer_Truthe.2015}}}(REG);
			\draw[hier] (MON) to node[pos=.45,edgeLabel]{\small\cite{Truthe.2018}} (NIL);
			\draw[hier] (NIL) to node[pos=.45,edgeLabel]{\small\cite{Wiedemann.1978}}(DEF);
			\draw[hier] (DEF) to node[pos=.45,edgeLabel]{\small\cite{Holzer_Truthe.2015}}(ORD);
			\draw[hier] (ORD) to node[pos=.45,edgeLabel]{\small\cite{Shyr_Thierrin.1974.ord}}(NC);
			\draw[hier] (NC) to node[pos=.4,edgeLabel]{\small\cite{Shyr_Thierrin.1974.ps}}(PS);
			\draw[hier] (PS) to node[pos=.45,edgeLabel]{\small\cite{Holzer_Truthe.2015}}(REG);
			\draw[hier, bend right] (MON) to node[pos=.7,edgeLabel]{\small \cite{Koedding.Truthe.2025.idefix}}(INF);
			\draw[hier, bend right] (SUF) to node[edgeLabel]{\small\cite{Holzer_Truthe.2015}}(PS);
			\draw[hier] (FIN) to node[pos=.45,edgeLabel]{\small\cite{Wiedemann.1978}}(NIL);
			\draw[hier] (COMB) to node[pos=.45,edgeLabel]{\small\cite{Havel.1969}}(DEF);
			\draw[hier] (COMB) to node[pos=.45,edgeLabel]{\small\cite{Olejar_Szabari.2023}}(SYDEF);
			\draw[hier] (SYDEF) to node[pos=.45,edgeLabel]{\small\cite{Koedding.Truthe.2024}}(LCOM);
			\draw[hier] (SYDEF) to node[pos=.45,edgeLabel]{\small\cite{Olejar_Szabari.2023}}(RCOM);
			\draw[hier] (SYDEF) to node[edgeLabel]{\small\cite{Olejar_Szabari.2023}}(PS);
			
			\draw[hier, bend left] (INF) to node[edgeLabel]{\small \cite{Koedding.Truthe.2025.idefix}}(PRE);
			\draw[hier] (INF) to node[edgeLabel]{\small \cite{Koedding.Truthe.2025.idefix}}(SUF);
			\draw[hier, bend right] (PRE) to node[edgeLabel]{\small \cite{Koedding.Truthe.2025.idefix}}(PS);
			
			\draw[hier] (INFF) to node[edgeLabel]{\small Le. \ref{lemma:inff_subset_pref_suf}}(PREF);
			\draw[hier, bend right=10] (INFF) to node[pos=0.75, edgeLabel]{\small Le. \ref{lemma:inff_subset_pref_suf}}(SUFF);
			\draw[hier, bend right =25] (PREF) to node[edgeLabel]{\small Le. \ref{lemma:pref_suff_subset_ps}}(PS);
			\draw[hier, bend right] (SUFF) to node[edgeLabel]{\small Le. \ref{lemma:pref_suff_subset_ps}}(PS);
	\end{tikzpicture}}
	\caption{Resulting hierarchy of subregular language families.}
	\label{fig:lang_erg_1}
\end{figure}
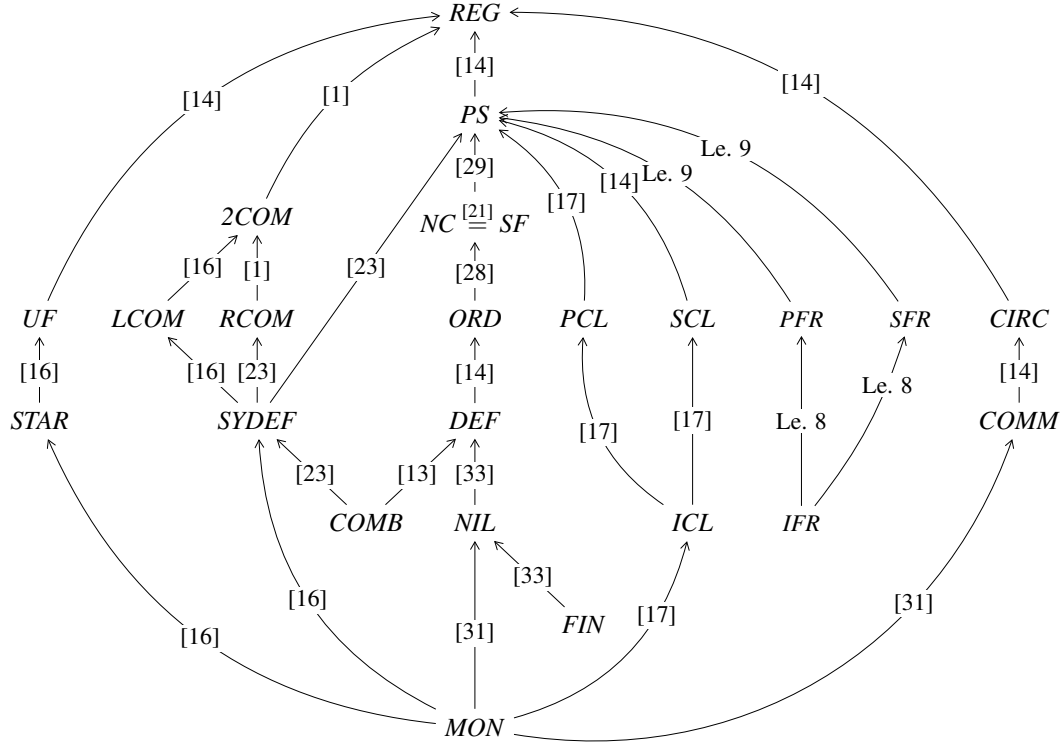

An arrow from a node $X$ to a node~$Y$ stands for the proper inclusion $X \subset Y$. 
  If two families are not connected by a directed path, then they are incomparable. 
  An edge label refers to the paper where the proper inclusion has been shown (in some cases, it might be that it is not the 
  first paper where the respective inclusion has been mentioned, since it is so obvious that it was not emphasized in a 
  publication) or the lemma of this paper where the proper inclusion will be shown.

  In the literature, it is often said that two languages are equivalent if they are equal or differ
  at most in the empty word. Similarly, two families can be regarded to be equivalent if they differ 
  only in the languages $\emptyset$ or $\{\lambda\}$. Therefore, the set $\STAR$ of all star languages is 
  sometimes regarded as a proper subset of the set $\COM$ of all (left-, right-, or two-sided) comet languages 
  although $\{\lambda\}$ belongs to the family $\STAR$ but not to $\LCOM$, $\RCOM$ 
  or $\TCOM$. We regard $\STAR$ and $\STAR\setminus\{\{\lambda\}\}$ as different.
  
  We now present some languages which will serve later as witness languages for proper inclusions or
  incomparabilities. 

\newcommand{\prefsetminussuff}{\ensuremath{\{b,ab\}}}
\begin{lemma}\label{lemma:pref_setminus_suff}
	Let $L = \{b, ab\}$. Then, it holds $L \in (\PREF \cap \FIN) \setminus \SUFF$.
\end{lemma}
\begin{proof}
	The language $L$ contains exactly two words and is therefore finite, yielding $L \in \FIN$. 
	The word $b$ is not a proper prefix of $ab$, and $ab$ is not a prefix of $b$. Since no word in $L$ is a proper prefix of another word in $L$, it follows that $L \in \PREF$. 
	However, since $b \in L$ is a proper suffix of $ab \in L$, the language $L$ is not suffix-free. Therefore, $L \notin \SUFF$, and the assertion holds.
\end{proof}

\newcommand{\suffsetminuspref}{\ensuremath{\{a,ab\}}}
\begin{lemma}\label{lemma:suff_setminus_pref}
	Let $L = \{a,ab\}$. Then, it holds $L \in (\SUFF \cap \FIN) \setminus \PREF$.
\end{lemma}
\begin{proof}
	The language $L$ is finite since it contains exactly two words, hence $L \in \FIN$. 
	The word $a$ is not a proper suffix of $ab$, and $ab$ is not a suffix of $a$. Since no word in $L$ is a proper suffix of another word in $L$, it holds $L \in \SUFF$. 
	Furthermore, since $a \in L$ is a proper prefix of $ab \in L$, the language $L$ is not prefix-free. Hence, $L \notin \PREF$, which yields the assertion.
\end{proof}
\pagebreak

\newcommand{\monsetminusprefcupsuff}{\ensuremath{\{a\}^*}}
\begin{lemma}\label{lemma:mon_setminus_pref_cup_suff}
	Let $L = \{a\}^*$. Then, it holds $L \in \MON \setminus (\PREF \cup \SUFF)$.
\end{lemma}
\begin{proof}
	With the alphabet $V = \{a\}$, the language $L$ can be expressed as $V^*$. Therefore, $L \in \MON$. 
	The word $a \in L$ is both a proper prefix and a proper suffix of the word $aa \in L$. Thus, the language $L$ is neither prefix-free nor suffix-free, which implies $L \notin \PREF \cup \SUFF$.
\end{proof}

\newcommand{\inffsetminusuf}{\ensuremath{\{ab, ba\}}}
\begin{lemma}\label{lemma:inff_setminus_uf}
	Let $L = \inffsetminusuf$. Then, it holds $L \in \INFF \setminus \UF$.
\end{lemma}
\begin{proof}
	Since the words $ab$ and $ba$ have the same length and are not equal, neither can be a proper infix of the other. Thus, $L \in \INFF$. 
	According to \cite{Nagy.2019}, a union-free language is either infinite or contains at most one word. Since $L$ contains exactly two words, it is not union-free. Therefore, $L \notin \UF$, which yields the assertion.
\end{proof}

\newcommand{\inffsetminustwocompresuf}{\ensuremath{\{ab\}}}
\begin{lemma}\label{lemma:inff_setminus_twocom_pre_suf}
	Let $L = \inffsetminustwocompresuf$. Then, it holds $L \in \INFF \setminus (\TCOM \cup \PRE \cup \SUF \cup \CIRC)$.
\end{lemma}
\begin{proof}
	The language $L$ consists of exactly one word and is therefore infix-free, yielding $L \in \INFF$. 
	According to \cite{Koedding.Truthe.2024}, every two-sided comet language is either empty or infinite. Since $L$ contains exactly one word, $L \notin \TCOM$. 
	Furthermore, $L$ is neither prefix-closed, suffix-closed nor circular because the word $ab$ is in $L$, but its proper prefix and proper suffix $\lambda$ and its circular permutation $ba$ is not. 
	Consequently, $L \notin \PRE \cup \SUF \cup \CIRC$.
\end{proof}

\newcommand{\inffsetminusnc}{\ensuremath{\{a(bb)^n a \mid n \ge 0\}}}
\begin{lemma}\label{lemma:inff_setminus_nc}
	Let $L = \inffsetminusnc$. Then, it holds $L \in \INFF \setminus \NC$.
\end{lemma}
\begin{proof}
	Every word in $L$ contains exactly two occurrences of the letter $a$, located at the very first and last positions. Because of this, no word in $L$ can be a proper infix of another word in $L$. Thus, $L \in \INFF$.
	
	Assuming that $L$ is non-counting, it follows from the definition that for all words $x, y, z \in \{a,b\}^*$ 
	and for a number $k \ge 1$ the equivalence $x y^k z \in L \iff x y^{k+1} z \in L$ applies. 
	We now set $x = z = a$ and $y = b$. If~$k$ is even, $a b^k a \in L$ but $a b^{k+1} a \notin L$, 
	which is a contradiction. If~$k$ is odd, $a b^{k+1} a \in L$ but $a b^k a \notin L$, which is also 
	a contradiction. It follows that $L \notin \NC$.
\end{proof}

\newcommand{\combsetminusprefcupsuff}{\ensuremath{\{a, b\}^*\{a\}}}
\begin{lemma}\label{lemma:comb_setminus_pref_cup_suff}
	Let $V = \{a, b\}$ and $L = \combsetminusprefcupsuff$. Then, it holds $L \in \COMB \setminus (\PREF \cup \SUFF)$.
\end{lemma}
\begin{proof}
	By setting $X = \{a\} \subseteq V$, the language can be written as $L = V^*X$. Therefore, $L \in \COMB$. 
	The word $a \in L$ is a proper prefix and a proper suffix of the word $aa \in L$, which implies that $L$ is neither prefix-free nor suffix-free. Consequently, $L \notin \PREF \cup \SUFF$, and the assertion holds.
\end{proof}

We now prove some inclusion relations.
\begin{lemma}\label{lemma:inff_subset_pref_suf}
	The proper inclusions $\INFF \subset \PREF$ and $\INFF \subset \SUFF$ hold.
\end{lemma}
\begin{proof}
	The inclusions $\INFF \subseteq \PREF$ and $\INFF \subseteq \SUFF$ hold because every prefix and every suffix of a word is also an infix of that word. Thus, any infix-free language is necessarily prefix-free and suffix-free.
	
	The language $L_1 = \prefsetminussuff$ from Lemma \ref{lemma:pref_setminus_suff} is a witness language for the properness of the first inclusion. Since $L_1 \in \PREF \setminus \SUFF$ and $\INFF \subseteq \SUFF$, it follows that $L_1 \in \PREF \setminus \INFF$.
	
	Similarly, the language $L_2 = \suffsetminuspref$ from Lemma \ref{lemma:suff_setminus_pref} is a witness language for the properness of the second inclusion. Since $L_2 \in \SUFF \setminus \PREF$ and $\INFF \subseteq \PREF$, it follows that $L_2 \in \SUFF \setminus \INFF$.
\end{proof}
\pagebreak

\begin{lemma}\label{lemma:pref_suff_subset_ps}
	The proper inclusions $\PREF \subset \PS$ and $\SUFF \subset \PS$ hold.
\end{lemma}
\begin{proof}
	We first show the inclusions $\PREF \subseteq \PS$ and $\SUFF \subseteq \PS$. Let $L$ be a regular language 
	over an alphabet $V$
	such that $L \in \PREF$ or $L \in \SUFF$. Since $L$ is regular, there exists a deterministic finite automaton accepting $L$ with $n$ states. We claim that $L$ is power-separating with the constant $m = n$. We show that for any $x \in V^*$, either $J_x^n \cap L = \emptyset$ or $J_x^n \subseteq L$, where $J_x^n = \set{x^k}{k \ge n}$.
	
	If $x = \lambda$, then $J_\lambda^n = \{\lambda\}$. If $\lambda \in L$, then $J_\lambda^n \subseteq L$; otherwise $J_\lambda^n \cap L = \emptyset$.
		
	For the case $x \neq \lambda$, we will show that $J_x^n \cap L = \emptyset$. 
	Let $x\in V^*$.
	Assume for the sake of contradiction that $J_x^n \cap L \neq \emptyset$. This means there exists an integer $k \ge n$ such that $x^k \in L$. Consider the states the DFA ended in after reading the prefixes $x^0, x^1, \dots, x^n$. Since there are $n+1$ states the DFA ended in but only $n$ distinct states, the Pigeonhole Principle dictates that at least two states among them are identical. Let $0 \leq i < j \leq n$ be two indices with 
	$\delta(q_0, x^i) = \delta(q_0, x^j)$ and set $d = j - i $ (note that $d\geq 1$). 
	Then $\delta(q_0, x^t) = \delta(q_0, x^{t+d})$ for all $t \geq i$. 
	Since $k \geq n > i$, we obtain 
	$\delta(q_0, x^{k+d}) = \delta(q_0, x^k) \in F$, so $x^{k+d} \in L$. 
	
	Since $x^{k+d} = x^k x^d = x^d x^k$ and $d \ge 1$, the word $x^k$ is both a proper prefix and a 
	proper suffix of~$x^{k+d}$. Having both $x^k$ and $x^{k+d}$ in $L$ contradicts the assumption 
	that $L$ is prefix-free or suffix-free. Thus, $J_x^n \cap L = \emptyset$ holds, which implies
	$\PREF \subseteq \PS$ and $\SUFF \subseteq \PS$.
	
	The language $L = \monsetminusprefcupsuff$ from Lemma \ref{lemma:mon_setminus_pref_cup_suff} is a witness for the properness of both inclusions. Since $\MON \subset \PS$ and $L \in \MON$, it holds $L \in \PS$. As $L \notin (\PREF \cup \SUFF)$, we obtain $L \in \PS \setminus \PREF$ and $L \in \PS \setminus \SUFF$.
\end{proof}

We now prove the incomparability relations mentioned in Figure~\ref{fig:lang_erg_1} which have not been proved earlier. 
These are the relations regarding the families $\PREF$, $\SUFF$, and $\INFF$. 

\begin{lemma}
	Let $\cF = \{\MON, \UF, \STAR\}$. Every family in $\cF$ is incomparable to the families $\PREF$, $\SUFF$, and $\INFF$.
\end{lemma}
\begin{proof}
	Let $F \in \cF$ and $X \in \{\PREF, \SUFF, \INFF\}$. Due to the inclusion relations of the subregular families, it suffices to show that there are languages
	\begin{align*}
		L_1 \in \MON \setminus (\PREF \cup \SUFF) \text{ and } L_2 \in \INFF \setminus \UF.
	\end{align*}
	
	For the first case, we refer to the 
	language $L_1 = \monsetminusprefcupsuff$ from Lemma \ref{lemma:mon_setminus_pref_cup_suff}. 
	As established there, it holds $L_1 \in \MON \setminus (\PREF \cup \SUFF)$. Since $\INFF \subseteq \PREF \cap \SUFF$, it also holds $L_1 \notin \INFF$, which implies $L_1 \notin X$. Since $\MON \subseteq \STAR \subseteq \UF$, we have $\MON \subseteq F$ for all $F \in \cF$. Thus, $L_1 \in F \setminus X$.
	
	For the second case, we consider the language $L_2 = \inffsetminusuf$ from Lemma \ref{lemma:inff_setminus_uf}. The lemma shows that $L_2 \in \INFF \setminus \UF$. Since $\INFF \subseteq \PREF \cap \SUFF$, it holds $L_2 \in X$. Since $\STAR \subseteq \UF$, the relation $L_2 \notin \UF$ implies $L_2 \notin F$ for all $F \in \cF$. Thus, $L_2 \in X \setminus F$.
	
	Since both directions of inclusion are refuted, the assertion holds.
\end{proof}

\begin{lemma}
	Let $\cF = \{\TCOM, \LCOM, \RCOM, \SYDEF, \COMB\}$. Every family in $\cF$ is incomparable to the families $\PREF$, $\SUFF$, and $\INFF$.
\end{lemma}
\begin{proof}
	Let $F \in \cF$ and $X \in \{\PREF, \SUFF, \INFF\}$. Due to the inclusion relations of the subregular families, it suffices to show that there are languages 
	\begin{align*}
		L_1 \in \COMB \setminus (\PREF \cup \SUFF) \text{ and } L_2 \in \INFF \setminus \TCOM.
	\end{align*}
	
	For the first case, we refer to the language $L_1 = \combsetminusprefcupsuff$ from Lemma \ref{lemma:comb_setminus_pref_cup_suff}. As established there, $L_1 \in \COMB \setminus (\PREF \cup \SUFF)$. Since $\INFF \subseteq \PREF \cap \SUFF$, it also holds $L_1 \notin \INFF$, which implies $L_1 \notin X$. Since $\COMB \subseteq F$ for all $F \in \cF$, we have $L_1 \in F$. Thus, $L_1 \in F \setminus X$.
	
	For the second case, we consider the language $L_2 = \inffsetminustwocompresuf$ from Lemma \ref{lemma:inff_setminus_twocom_pre_suf}. The lemma shows that $L_2 \in \INFF$, which implies $L_2 \in X$, and that $L_2 \notin \TCOM$. Since all families $F \in \cF$ are subsets of $\TCOM$, the relation $L_2 \notin \TCOM$ implies $L_2 \notin F$. Thus, $L_2 \in X \setminus F$.
	
	Since both directions of inclusion are refuted, the assertion holds.
\end{proof}

\begin{lemma}
	Let $\cF = \{\NC, \ORD, \DEF, \NIL, \FIN\}$. Every family in $\cF$ is incomparable to the families $\PREF$, $\SUFF$, and $\INFF$.
\end{lemma}
\begin{proof}
	Let $F \in \cF$ and $X \in \{\PREF, \SUFF, \INFF\}$. Due to the inclusion relations of the subregular families, it suffices to show that there are languages
	\begin{align*}
		L_1 \in \FIN \setminus (\PREF \cup \SUFF) \text{ and } L_2 \in \INFF \setminus \NC.
	\end{align*}
	
	For the first case, we consider the language $L_1 = \{a, aa\}$. Since $L_1$ is finite, $L_1 \in \FIN$. The language $L_1$ is neither prefix-free nor suffix-free, yielding $L_1 \notin \PREF \cup \SUFF$. Since $\INFF \subseteq \PREF \cap \SUFF$, it also holds $L_1 \notin \INFF$, which implies $L_1 \notin X$. Due to the inclusions $\FIN \subseteq \NIL \subseteq \DEF \subseteq \ORD \subseteq \NC$, we have $\FIN \subseteq F$ for all $F \in \cF$, and thus $L_1 \in F$. Therefore, $L_1 \in F \setminus X$.
	
	For the second case, we refer to the language $L_2 = \inffsetminusnc$ from Lemma \ref{lemma:inff_setminus_nc}. 
	The lemma shows that $L_2 \in \INFF \setminus \NC$. Since $\INFF \subseteq \PREF \cap \SUFF$, it holds $L_2 \in X$. 
	Since all families $F \in \cF$ are subsets of $\NC$, the relation $L_2 \notin \NC$ implies $L_2 \notin F$. Thus, $L_2 \in X \setminus F$.
	
	Since both directions of inclusion are refuted, the assertion holds.
\end{proof}

\begin{lemma}
	Let $\cF = \{\PRE, \SUF, \INF\}$. Every family in $\cF$ is incomparable to the families $\PREF, \SUFF$, and~$\INFF$.
\end{lemma}
\begin{proof}
	Let $F \in \cF$ and $X \in \{\PREF, \SUFF, \INFF\}$. Due to the inclusion relations of the subregular families, it suffices to show that there are languages
	\begin{align*}
		L_1 \in \INF \setminus (\PREF \cup \SUFF) \text{ and } L_2 \in \INFF \setminus (\PRE \cup \SUF).
	\end{align*}
	For the first case, we refer to the language $L_1 = \monsetminusprefcupsuff$ from 
	Lemma \ref{lemma:mon_setminus_pref_cup_suff}. As established there, it holds $L_1 \in \MON \setminus (\PREF \cup \SUFF)$. Since $\INFF \subseteq (\PREF \cap \SUFF)$, it also holds $L_1 \notin \INFF$, which implies $L_1 \notin X$. Since $\MON \subseteq \INF \subseteq (\PRE \cap \SUF)$, we have $\MON \subseteq F$ for all $F \in \cF$, yielding $L_1 \in F$. Thus, $L_1 \in F \setminus X$.
	
	For the second case, we consider the language $L_2 = \inffsetminustwocompresuf$ from Lemma \ref{lemma:inff_setminus_twocom_pre_suf}. The lemma shows that $L_2 \in \INFF$, meaning $L_2 \in X$, and that $L_2 \notin \PRE \cup \SUF$. Since closure under subwords would require the empty word $\lambda$ to be in the language, which is false for $L_2$, we also have $L_2 \notin \INF$. Thus, $L_2 \notin F$ for all $F \in \cF$, yielding $L_2 \in X \setminus F$.
	
	Since both directions of inclusion are refuted, the assertion holds.
\end{proof}
\pagebreak

\begin{lemma}
	Let $\cF = \{\COMM, \CIRC\}$. Every family in $\cF$ is incomparable to the families $\PREF, \SUFF$, and~$\INFF$.
\end{lemma}
\begin{proof}
	Let $F \in \cF$ and $X \in \{\PREF, \SUFF, \INFF\}$. Due to the inclusion relations of the subregular families, it suffices to show that there are languages
	\begin{align*}
		L_1 \in \MON \setminus (\PREF \cup \SUFF) \text{ and } L_2 \in \INFF \setminus \CIRC.
	\end{align*}
	For the first case, we refer to the language $L_1 = \monsetminusprefcupsuff$ from 
	Lemma \ref{lemma:mon_setminus_pref_cup_suff}. As established there, it holds $L_1 \in \MON \setminus (\PREF \cup \SUFF)$. Since $\INFF \subseteq (\PREF \cap \SUFF)$, it also holds $L_1 \notin \INFF$, which implies $L_1 \notin X$. Since $\MON \subseteq \COMM$ and $\MON \subseteq \CIRC$, we have $\MON \subseteq F$ for all $F \in \cF$, yielding $L_1 \in F$. Thus, $L_1 \in F \setminus X$.
	
	For the second case, we consider the language $L_2 = \{ab\}$. No word in $L_2$ is a proper infix of another, meaning $L_2 \in \INFF$, which implies $L_2 \in X$. However, $L_2$ lacks the word $ba$, violating the closure properties of commutativity and circularity, so $L_2 \notin (\COMM \cup \CIRC)$. Furthermore, it holds that $L_2 \notin \MON$. Thus, $L_2 \notin F$ for all $F \in \cF$, yielding $L_2 \in X \setminus F$.
	
	Since both directions of inclusion are refuted, the assertion holds.
\end{proof}

From all these relations, the hierarchy presented in Figure~\ref{fig:lang_erg_1} follows.

\begin{theorem}[Resulting hierarchy for subregular families]\label{theorem:neue_hierarchie}
The inclusion relations presented in Figure \ref{fig:lang_erg_1} hold. An arrow from an entry $X$ to
an entry~$Y$ depicts the proper inclusion $X \subset Y$; if two families are not connected by a directed
path, they are incomparable.
\end{theorem}

\section{Results on subregular control in external contextual grammars}

In this section, we include the families of languages generated by external contextual grammars with 
selection languages from the subregular families under investigation into the existing hierarchy 
with respect to external contextual grammars.

If, in a contextual grammar, all selection languages belong to some language family~$X$, then they belong also
to every super set $Y$ of $X$. Therefore, each language in~$\ec{X}$ is also generated by a contextual grammar
with selection languages from $Y$ and we have the following monotonicity.

\begin{lemma}[Monotonicity $\cEC$]\label{lemma:ec_monoton}
	For any two language classes $X$ and $Y$ with~$X\subseteq Y$,
	we have the inclusion~$\cEC(X)\subseteq\cEC(Y)$.
\end{lemma}

Figure~\ref{fig:ec_erg} shows 
the inclusion relations between language families 
which are generated by external
contextual grammars where the selection languages belong to subregular classes investigated before. The 
hierarchy contains results which were already known (marked by a reference to the literature) and results which 
will be proved in this section (marked by a number which refers to the respective lemma).

\begin{figure}[htb]
	\centering
		\scalebox{.77}{
			\begin{tikzpicture}[node distance=15mm and 25mm,on grid=true
				]
				\node (MON) {$\ec{\MON}$};
				\node (FIN) [below=of MON] {$\ec{\FIN} \stackrel{\text{ Le. \ref{cor:ec_inff_eq_fin}}}{=} \ec{\INFF}$};
				\node (COMB)[above=of MON] {$\ec{\COMB}$};
				\node (NIL) [right=of COMB] (NIL) {$\ec{\NIL}$};
				\node (DEF) [above=of COMB] {$\ec{\DEF}$};
				\node (ORD) [above=of DEF] {$\ec{\ORD}$};
				\node (SYDEF) [right=of ORD] {$\ec{\SYDEF}$};
				\node (INF) at (-3.7, 3) {$\ec{\INF}$};
				\node (SUF) [left=of ORD] {$\ec{\SUF}$};
				\node (PRE) [left=of SUF] {$\ec{\PRE}$};
				
				\node (PREF) [right=of SYDEF] {$\ec{\PREF}$};
				\node (SUFF) [right=of PREF] {$\ec{\SUFF}$};
				
				\node (COMM) [right=of SUFF] {$\ec{\COMM}$};

				\node (STAR) [left=of PRE] {$\ec{\STAR}$};
				\node (NC) [above=of ORD] {$\ec{\NC}$};
				\node (PS) [above=of NC] {$\ec{\PS}$};
				\node (CIRC) [above=of COMM] {$\ec{\CIRC}$};
				\node (REG) [above=of PS] {$\ec{\REG} \stackrel{\text{\cite{Dassow_Manea_Truthe.2012}}}{=} \ec{\UF} \stackrel{\text{\cite{Koedding.Truthe.2024}}}{=} \cEC(\LCOM) \stackrel{\text{\cite{Koedding.Truthe.2024}}}{=} \cEC(\RCOM) \stackrel{\text{\cite{Koedding.Truthe.2024}}}{=} \cEC(\TCOM)$};
				
				\draw[hier] (FIN) to node[pos=.45, edgeLabel]{\small\cite{Dassow.2005}} (MON);
				\draw[hier, bend right] (MON) to node[pos=.45, edgeLabel]{\small\cite{Dassow.2005}} (NIL);
				\draw[hier] (MON) to node[pos=.45, edgeLabel]{\small\cite{Dassow.2015}} (COMB);
				\draw[hier] (COMB) to node[pos=.45, edgeLabel]{\small\cite{Truthe.2021}} (DEF);
				\draw[hier] (DEF) to node[pos=.45, edgeLabel]{\small\cite{Truthe.2014}} (ORD);
				\draw[hier, bend right=15] (DEF) to node[pos=.45, edgeLabel]{\small\cite{Koedding.Truthe.2024}} (SYDEF);
				\draw[hier] (ORD) to node[pos=.45, edgeLabel]{\small\cite{Dassow_Truthe.2023}} (NC);
				\draw[hier] (NC) to node[pos=.45, edgeLabel]{\small\cite{Truthe.2021}} (PS);
				\draw[hier, bend left=30] (MON) to node[pos=.45, edgeLabel]{\small\cite{Koedding.Truthe.2025.idefix}}
				(INF);
				\draw[hier, bend right=30] (INF) to node[pos=.45, edgeLabel]{\small\cite{Koedding.Truthe.2025.idefix}}
				(SUF);
				\draw[hier, bend left=30] (INF) to node[pos=.45, edgeLabel]{\small\cite{Koedding.Truthe.2025.idefix}}
				(PRE);
				\draw[hier, bend left=40] (SUF) to node[pos=.45, edgeLabel]{\small\cite{Truthe.2021}} (PS);
				\draw[hier, bend left=30] (PRE) to node[pos=.45, edgeLabel]{\small\cite{Koedding.Truthe.2025.idefix}}
				(PS);
				\draw[hier, bend left=35] (MON) to node[pos=.45, edgeLabel]{\small\cite{Koedding.Truthe.2024}} (STAR);
				\draw[hier, bend left=25] (STAR) to node[pos=.45, edgeLabel]{\small\cite{Koedding.Truthe.2024}} (REG);
				\draw[hier, bend right] (SYDEF) to node[pos=.45, edgeLabel]{\small\cite{Koedding.Truthe.2024}} (PS);
				\draw[hier, bend right] (NIL) to node[pos=.65, edgeLabel]{\small\cite{Dassow.2005}} (COMM);
				\draw[hier, bend right] (NIL) to node[pos=.45, edgeLabel]{\small \cite{Dassow.2005}} (DEF);
				\draw[hier] (COMM) to node[pos=.45, edgeLabel]{\small\cite{Dassow_Manea_Truthe.2012}} (CIRC);
				\draw[hier, bend right=15] (CIRC) to node[pos=.45, edgeLabel]{\small\cite{Dassow_Manea_Truthe.2012}} (REG);
				\draw[hier] (PS) to node[pos=.45, edgeLabel]{\small\cite{Truthe.2021}} (REG);
				\draw[hier, bend right] (FIN) to node[pos=.45, edgeLabel]{\small Le. \ref{cor:ec_inclusions}} (PREF);
				\draw[hier, bend right] (FIN) to node[pos=.45, edgeLabel]{\small Le. \ref{cor:ec_inclusions}} (SUFF);
				\draw[hier, bend right=22] (PREF) to node[pos=.45, edgeLabel]{\small Le. \ref{lemma:ec_pref_suff_subset_ps}} (PS);
				\draw[hier, bend right=25] (SUFF) to node[pos=.45, edgeLabel]{\small Le. \ref{lemma:ec_pref_suff_subset_ps}} (PS);
			\end{tikzpicture}
		}
		\caption{Resulting hierarchy of language families by external contextual grammars with special selection languages.}
		\label{fig:ec_erg}
	\end{figure}
	
	We now present some languages which will 
	serve later 
	as witness languages for proper inclusions or
	incomparabilities.
\pagebreak

\begin{lemma}\label{lemma:ec_pref_setminus_suff}
	Let $L = \{a\}^*\{b\}$. Then, it holds $L \in \ec{\PREF} \setminus \ec{\SUFF}$.
\end{lemma}
\begin{proof}
	The external contextual grammar $G = (\{a,b\}, \{\{a\}^*\{b\} \to (a, \lambda)\}, \{b\})$ generates the language~$L$. 
	For any integers $m$ and $n$ with $m < n$, the word $a^m b$ is not a prefix of $a^n b$ since the letter $b$ appears 
	exactly once at the very end of each word. Thus, the selection language $\{a\}^*\{b\}$ is prefix-free, 
	yielding~$L_1 \in \ec{\PREF}$. 
	
	Assume $L \in \ec{\SUFF}$. Any grammar generating the infinite language $L_1$ from a finite set of axioms 
	must have at least one selection language $S \in \SUFF$ containing infinitely many words of $L$. Thus, 
	there exist integers $x$ and $y$ with $x < y$ such that $a^x b \in S$ and $a^y b \in S$. 
	Since $a^y b = a^{y-x} a^x b$ and $a^{y-x} \neq \lambda$, the word~$a^x b$ is a proper suffix of $a^y b$. 
	This contradicts the assumption that $S \in \SUFF$. Therefore, $L_1 \notin \ec{\SUFF}$, which yields the assertion.
\end{proof}

\begin{lemma}\label{lemma:ec_suff_setminus_pref}
	Let $L = \{b\}\{a\}^*$. Then, it holds $L \in \ec{\SUFF} \setminus \ec{\PREF}$.
\end{lemma}
\begin{proof}
	The external contextual grammar $G = (\{a,b\}, \{\{b\}\{a\}^* \to (\lambda, a)\}, \{b\})$ generates the language~$L$. For any integers $m$ and $n$ with $m < n$, the word $b a^m$ is not a suffix of $b a^n$ since the letter $b$ appears exactly once at the very beginning of each word. Thus, the selection language $\{b\}\{a\}^*$ is suffix-free, yielding $L \in \ec{\SUFF}$.
	
	Assume $L \in \ec{\PREF}$. A grammar generating the infinite language $L$ must have a selection language $S \in \PREF$ containing infinitely many words of $L$. Thus, there exist integers $x$ and $y$ with $x < y$ such that $b a^x \in S$ and $b a^y \in S$. Since $b a^y = b a^x a^{y-x}$ and $a^{y-x} \neq \lambda$, the word $b a^x$ is a proper prefix of $b a^y$. This contradicts the assumption that $S \in \PREF$. Therefore, $L \notin \ec{\PREF}$, establishing the assertion.
\end{proof}
\pagebreak

\begin{lemma}\label{lemma:ec_mon_setminus_pref_suff}
	Let $L = \{a\}^*$. Then, it holds $L \in \ec{\MON} \setminus (\ec{\PREF} \cup \ec{\SUFF})$.
\end{lemma}
\begin{proof}
	The external contextual grammar $G = (\{a\}, \{\{a\}^* \to (a, \lambda)\}, \{\lambda\})$ generates the language $L$. Since the only selection language is $\{a\}^* \in \MON$, it holds $L \in \ec{\MON}$.
	
	Assuming $L \in \ec{\PREF} \cup \ec{\SUFF}$, the language $L$ is generated by an external contextual grammar where all selection languages belong entirely to $\PREF$ or entirely to $\SUFF$. Since $L$ is infinite and the grammar has only finitely many selection pairs, at least one selection language $S$ must contain infinitely many words from $L$. Thus, there exist integers $n$ and $m$ with $n < m$ such that $a^n, a^m \in S$. Since $a^n$ is both a proper prefix and a proper suffix of $a^m$, the selection language $S$ is neither prefix-free nor suffix-free. This is a contradiction. Consequently, $L \notin \ec{\PREF} \cup \ec{\SUFF}$, which yields the assertion.
\end{proof}

\begin{lemma}\label{lemma:ec_pref_suff_setminus_pre_suf_inf}
	Let $L_\textit{pre} = \{a, b\} \cup \set{a^n c b}{n \ge 1}$ and $L_\textit{suf} = L_\textit{pre}^R$. Then, it holds 
	\[L_\textit{pre} \in \ec{\PREF} \setminus (\ec{\PRE} \cup \ec{\SUF})\] 
	and 
	\[L_\textit{suf} \in \ec{\SUFF} \setminus (\ec{\PRE} \cup \ec{\SUF}).\]
\end{lemma}
\begin{proof}
	The external contextual grammar $G = (\{a,b,c\}, \{\{a\}^* \{c b\} \to (a, \lambda)\}, \{a, b, acb\})$ generates the language $L_{\textit{pre}}$. For any $n \neq m$, the word $a^n c b$ is not a prefix of $a^m c b$. Thus, the only selection language $\{a\}^* \{cb\}$ is prefix-free, yielding $L_{\textit{pre}} \in \ec{\PREF}$.
	
	Assume $L_{\textit{pre}} \in \ec{\PRE} \cup \ec{\SUF}$. Any grammar generating the infinite 
	language $L_{\textit{pre}}$ from a finite set of axioms must have a selection language $S$ 
	containing infinitely many words of the form $a^n c b$. The corresponding context must be 
	of the form $(a^p, \lambda)$ with $p \ge 1$ to 
	generate longer words in $L_{\textit{pre}}$. 
	If $S \in \PRE$, the proper prefix $a$ of $a^n c b$ belongs to $S$. Since $a \in L_{\textit{pre}}$, 
	the grammar applies the context $(a^p, \lambda)$ to $a$, generating the word $a^{p+1} \notin L_{\textit{pre}}$, 
	a contradiction. 
	If $S \in \SUF$, the proper suffix $b$ belongs to $S$. Applying the context $(a^p, \lambda)$ to the 
	word $b \in L_{\textit{pre}}$ generates $a^p b \notin L_{\textit{pre}}$, which is also a contradiction. 
	Since $\INF \subseteq \PRE \cap \SUF$, we obtain $L_{\textit{pre}} \notin \ec{\PRE} \cup \ec{\SUF} \cup \ec{\INF}$.
	
	By symmetry, the reversed language $L_{\textit{suf}} = L_{\textit{pre}}^R$ is generated by an external 
	contextual grammar with the suffix-free selection language $\{b c\}\{a\}^*$ and the context $(\lambda, a)$, 
	yielding $L_{\textit{suf}} \in \ec{\SUFF}$. Assuming $L_{\textit{suf}} \in \ec{\PRE} \cup \ec{\SUF}$ leads 
	to analogous contradictions: a prefix-closed selection language would contain $b$, 
	triggering the generation of the word $b a^p \notin L_{\textit{suf}}$, while a suffix-closed selection 
	language would contain $a$, triggering $a^{p+1} \notin L_{\textit{suf}}$. 
	Thus, $L_{\textit{suf}} \in \ec{\SUFF} \setminus (\ec{\PRE} \cup \ec{\SUF} \cup \ec{\INF})$.
\end{proof}

\begin{lemma}\label{lemma:ec_pref_suff_setminus_nc}
	Let \begin{align*}
		L_1 &= \{a^mbc^{2n}db\mid m \geq 0, n \geq 1\},\\
		L_2 &= \{c^nd\mid n \geq 2\},\\
		L_3 &= \{bc^ndb\mid n \geq 2\},\\
		L_\textit{pre} &= L_1 \cup L_2 \cup L_3 \qmand L_\textit{suf} = L_\textit{pre}^R.
	\end{align*}
	Then, it holds $L_{\textit{pre}} \in \ec{\PREF} \setminus \ec{\NC}$  and $L_{\textit{suf}} \in \ec{\SUFF} \setminus \ec{\NC}$.
\end{lemma}
\begin{proof}
	First, we show that $L_{\textit{pre}} \in \ec{\PREF} \setminus \ec{\NC}$.
	
	The contextual grammar $G = (\{a,b,c,d\}, \{S_1 \to C_1, S_2 \to C_2\}, \{ccd\})$ with
	\iirule
	{\{c\}^*\{d\}}
	{\{(c, \lambda), (b, b)\}}
	{\{ a^p b c^{2n} d b \mid n \geq 1, p \geq 0 \}}
	{\{(a, \lambda)\}}
	generates the language $L_{\textit{pre}}$ and all selection languages are prefix-free. 
	The selection language $S_1$ is prefix-free since every word ends with its only letter $d$. Similarly, $S_2$ is prefix-free because all its words end with a $b$ that follows the only $d$ in the word, meaning no word in $S_2$ can be a proper prefix of another word in $S_2$.
	
	It can be seen as follows that the language $L_{\textit{pre}}$ is generated. The axiom $ccd$ 
	is in the first selection language $S_1$. If we apply the first rule $(c, \lambda)$ of $C_1$ arbitrarily often, 
	we get the words of the form $c^n d$ for $n \geq 2$. If we apply the second context $(b, b)$ of the first selection component to these words, we get the words $b c^n d b$ for $n \geq 2$. As soon as we use this context, our word begins with $b$ and is no longer in the first selection language.
	
	If the number of letters $c$ in our word is even, meaning the word is of the form $b c^{2n} d b$, it belongs to the second selection language $S_2$. We can then repeatedly apply the rule of the second selection component $(a, \lambda)$ to derive any word of the form $a^m b c^{2n} d b$ for $m \geq 1$. No other words are created in the process. Words with an odd number of $c$'s do not trigger the context in $S_2$. Thus, we see that the grammar generates the language $L_{\textit{pre}}$.
	
	We assume that $L_{\textit{pre}} \in \ec{\NC}$ holds. Then, there is a contextual 
	grammar $G' = (\{a,b,c,d\}, \mathcal{S}', B')$ where all selection languages are non-counting and 
	$L_{\textit{pre}} = L(G')$. 
	
	Since every selection language $S$ is non-counting, there is a natural number $k$ for each language such that, for all words $x, y, z \in V^*$, it holds $x y^k z \in S \iff x y^{k+1} z \in S$. We denote our selection languages by $S^{(k)}$ where $k$ is the smallest natural number in the sense of the definition of non-counting languages.
	
	Furthermore, we define $p = \max\{ k \mid (S^{(k)}, C) \in \mathcal{S}' \}$. Thus, the following statement applies to every selection language $S^{(k)}$: For all words $x, y, z \in V^*$, it holds
	\begin{equation}
		x y^p z = x y^k y^{p-k} z \in S^{(k)} \iff x y^{k+1} y^{p-k} z = x y^{p+1} z \in S^{(k)}. \label{eq:nc_prop}
	\end{equation}
	
	Since the language $L_1$ contains words with an arbitrary even number of letters $c$, there is a derivation
	$$w_0 \stackrel{*}{\implies} w_1 \implies u w_1 v$$
	with $w_0 \in B'$, $u w_1 v \in L_1$, $|w_1|_c > p + \ell(G')$, $|w_1|_a = 0$ and $|uv|_a > 0$. We now distinguish two cases. In the first case, the word $w_1$ begins with letter $b$; in the second case, the word $w_1$ begins with letter $c$.
	
	\textit{Case 1 ($w_1$ starts with $b$):}
	In this case, we have $w_1 = b c^k d b$ with $k > p + \ell(G')$. Since $u w_1 v$ belongs to the language $L_1$, $k$ is even.
	Let $S$ be the selection language used in the derivation step $w_1 \implies u w_1 v$. Since $w_1 \in S$, we obtain according to the relation (\ref{eq:nc_prop})
	$$w_1 = b c^p c^{k-p} d b \in S \iff b c^{p+1} c^{k-p} d b = b c^{k+1} d b \in S.$$
	Since $b c^{k+1} d b$ is in $L_3$ and, thus, also in $L_{\textit{pre}}$, the word $u b c^{k+1} d b v$ is also derived. However, since $k + 1$ is odd and $|uv|_a > 0$, this word is not in $L_{\textit{pre}}$, which is a contradiction to $L_{\textit{pre}} = L(G')$.
	
	\textit{Case 2 ($w_1$ starts with $c$):}
	Then, $w_1 = c^k d$ follows with $k > p + \ell(G')$. Since $u w_1 v \in L_1$ holds, it follows that $u = a^m b$ for a number $m \geq 1$ and $v = b$ (since $|uv|_a > 0$).
	Let $S$ denote the selection language used in the derivation step $w_1 \implies u w_1 v$. Since $w_1 \in S$, we obtain according to the relation (\ref{eq:nc_prop})
	$$w_1 = c^p c^{k-p} d \in S \iff c^{p+1} c^{k-p} d = c^{k+1} d \in S.$$
	Since $c^{k+1} d$ belongs to $L_2$ and, thus, to $L_{\textit{pre}}$, the word $u c^{k+1} d v$ is also derived. 
	Since $|uv|_a > 0$ and the word $u c^k d v$ belongs to $L_1$, the number $k$ is even and $k+1$ is an odd number. 
	Consequently, the derived word $a^m b c^{k+1} d b$ does not belong to the language $L_1$ and also not 
	to $L_{\textit{pre}}$, which is a contradiction to~$L_{\textit{pre}} = L(G')$.
	
	Since all cases lead to a contradiction, the assumption that $L_{\textit{pre}}$ is in $\ec{\NC}$ is false.
	
	By symmetry, an analogous argument holds for the reversed language $L_{\textit{suf}} = L_{\textit{pre}}^R$. 
	A contextual grammar with symmetrically defined suffix-free selection languages generates $L_{\textit{suf}}$. 
	Assuming $L_{\textit{suf}} \in \ec{\NC}$ yields a similar contradiction when pumping the non-counting sequences 
	of letters $c$. Therefore, we also have $L_{\textit{suf}} \in \ec{\SUFF} \setminus \ec{\NC}$.
\end{proof}

\begin{lemma}\label{lemma:ec_pref_suff_not_in_circ_comm}
	Let \[L_{\textit{pre}} = \{ (ab)^n c \mid n \ge 2 \} \cup \{ (ba)^n b c a \mid n \ge 1 \} \text{ and } L_\textit{suf} = L_\textit{pre}^R.\] 
	Then, it holds $L_\textit{pre} \in \ec{\PREF}\setminus \ec{\CIRC}$ and $L_\textit{suf} \in \ec{\SUFF}\setminus \ec{\CIRC}$.
\end{lemma}
\begin{proof}
	The external contextual grammar $G = (\{a,b,c\}, \{S_1 \to C_1, S_2\to C_2\}, \{ababc, babca\})$
	with
	\iirule
	{\set{(ab)^{n} c}{n \ge 2}}
	{\{\left(ab, \lambda\right)\}}
	{\set{(ba)^{n} b c a}{n \ge 1}}
	{\{(ba, \lambda)\}}
	generates the language $L_{\textit{pre}}$. 
	Every word in the first selection language $\set{(ab)^{n} c}{n \ge 2}$ ends with $c$, whereas any of its proper prefixes ends with $a$ or $b$. Thus, it is prefix-free. Similarly, every word in the second selection language $\set{(ba)^{n} b c a}{n \ge 1}$ ends with the subword $ca$. Since the letter $c$ occurs exactly once in each word, no proper prefix can end with $ca$. Therefore, both selection languages are prefix-free, yielding $L_{\textit{pre}} \in \ec{\PREF}$.
	
	Assume $L_{\textit{pre}} \in \ec{\CIRC}$. Any grammar generating $L_{\textit{pre}}$ must use a circular 
	selection language $S \in \CIRC$ containing a word $w_1 = (ab)^n c \in L_{\textit{pre}}$ (with $n \ge 2$) 
	and apply a context $((ab)^m, \lambda)$ to it to derive the longer word $(ab)^{n+m} c$ for a natural 
	number $m \geq 1$. Since $S$ is closed under circular shifts, shifting the first letter $a$ to the 
	end of $w_1$ yields the word $(ba)^{n-1} b c a$, which also belongs to $S$. Since this shifted word 
	is a word in $L_{\textit{pre}}$, 
	the context $((ab)^m, \lambda)$ can be applied to this word, too. This generates the 
	word $(ab)^m(ba)^{n-1} b c a$, which contains the subword $bb$. However, all words in $L_{\textit{pre}}$ 
	consist of alternating letters $a$ and $b$ before the letter $c$. This contradiction 
	implies $L_{\textit{pre}} \notin \ec{\CIRC}$.
	
	By symmetry, the reversed language $L_{\textit{suf}} = L_{\textit{pre}}^R$ is generated by an external 
	contextual grammar with the suffix-free selection languages $\set{c(ba)^{n}}{n \ge 2}$ 
	and $\set{acb(ab)^{n}}{n \ge 1}$, yielding $L_{\textit{suf}} \in \ec{\SUFF}$. 
	Assuming $L_{\textit{suf}} \in \ec{\CIRC}$ forces the application of the context $(\lambda, ba)$ 
	to the circularly shifted \linebreak
	word~$acb(ab)^{n-1} \in L_{\textit{suf}}$, generating a word with the 
	subword $bb$, which is a contradiction. Thus, $L_{\textit{suf}} \in \ec{\SUFF} \setminus \ec{\CIRC}$.
\end{proof}

\begin{lemma}\label{lemma:ec_pref_not_in_star}
	Let $L_{\textit{pre}} = \{\lambda\} \cup \{ a^n b c \mid n \ge 1 \}$ and $L_\textit{suf} = L_\textit{pre}^R$. Then, it holds $L_\textit{pre} \in \ec{\PREF} \setminus \ec{\STAR}$ and $L_\textit{suf} \in \ec{\SUFF} \setminus \ec{\STAR}$.
\end{lemma}
\begin{proof}
	The external contextual grammar $G = (\{a,b,c\}, \{\set{a^n b c}{n \ge 1} \to (a, \lambda)\}, \{\lambda, abc\})$ 
	generates the language $L_{\textit{pre}}$. Every word in the selection language $\set{a^n b c}{n\ge 1}$ ends with $c$, while any proper prefix ends with $a$ or $b$. Thus, the selection language is prefix-free, yielding $L_{\textit{pre}} \in \ec{\PREF}$.
	
	Assume $L_{\textit{pre}} \in \ec{\STAR}$. Any grammar generating the infinite language $L_{\textit{pre}}$ 
	must use a star selection language $S \in \STAR$ to extend words of the form $a^n b c$. 
	To generate longer words in $L_{\textit{pre}}$, the corresponding context must be $(a^k, \lambda)$ 
	for some $k \ge 1$. By definition of the Kleene star, every star language contains the empty word, 
	so $\lambda \in S$. Furthermore, since $\lambda \in L_{\textit{pre}}$, the context $(a^k, \lambda)$ 
	can be applied to this word, generating the 
	word $a^k$. Since $k \ge 1$, it holds $a^k \neq \lambda$ and $a^k \notin L_{\textit{pre}}$, which is a contradiction.
	Therefore, $L_{\textit{pre}} \notin \ec{\STAR}$.
	
	By symmetry, the reversed language $L_{\textit{suf}} = L_{\textit{pre}}^R$ is generated by an external 
	contextual grammar with the suffix-free selection language $\set{cba^n}{n \ge 1}$ and the context $(\lambda, a)$, 
	yielding $L_{\textit{suf}} \in \ec{\SUFF}$. Assuming $L_{\textit{suf}} \in \ec{\STAR}$ leads to a similar 
	contradiction: a selection language $S \in \STAR$ contains $\lambda$, allowing the application of the 
	corresponding context $(\lambda, a^k)$ to the axiom $\lambda$. 
	This generates the word $a^k \notin L_{\textit{suf}}$. Thus, $L_{\textit{suf}} \in \ec{\SUFF} \setminus \ec{\STAR}$.
\end{proof}

\begin{lemma}\label{lemma:ec_pref_not_in_sydef}
	Let $L_{\textit{pre}} = \{ a^n d \mid n \ge 0 \} \cup \{ c a^n d \mid n \ge 0 \} \cup \{ a^n d a^m d \mid n, m \ge 0 \}$.
	Then, it holds 
	\[L_{\textit{pre}} \in \ec{\PREF} \setminus \ec{\SYDEF}.\]
\end{lemma}
\begin{proof}
	First, we show that $L_{\textit{pre}} \in \ec{\PREF}$. The external contextual grammar $G = (\{a,c,d\}, \cS, \{d\})$ with the selection pairs $\cS = \{(S_1, C_1), (S_2, C_2)\}$ defined as
	\iirule
	{\set{a^n d}{n \ge 0}}
	{\{(a, \lambda), (c, \lambda), (d, \lambda)\}}
	{\set{a^m d a^n d}{m\ge 0,\ n\ge 0}}
	{\{(a, \lambda)\}}
	generates $L_{\textit{pre}}$. Every word in $S_1$ ends with exactly one $d$. Every word in $S_2$ contains exactly two occurrences of the letter $d$ and ends with the second one. Thus, no word in $S_1$ or $S_2$ can be a proper prefix of another word within the same selection language, yielding $S_1, S_2 \in \PREF$.
	
	Starting from the axiom $d$, applying $(a, \lambda)$ from $C_1$ generates the base set $a^* d$. 
	From any word $a^n d \in S_1$, applying $(c, \lambda)$ generates $c a^n d$. Since these words start with $c$, 
	they neither belong to  $S_1$ nor $S_2$, halting their derivation. Applying $(d, \lambda)$ from $C_1$ to $a^n d$ 
	generates $d a^n d$, which belongs to $S_2$. Applying $(a, \lambda)$ from $C_2$ iteratively to words in $S_2$ 
	generates $a^m d a^n d$. Since these words have two letters $d$, they do not belong to $S_1$. 
	Thus, $G$ generates $L_{\textit{pre}}$, yielding $L_{\textit{pre}} \in \ec{\PREF}$.
	
	Assume that $L_{\textit{pre}} \in \ec{\SYDEF}$ via some grammar $G' = (V, \cS', A')$. In order to generate the infinite subset $\{ c a^n d \mid n \ge 0 \}$, $G'$ must apply contexts to shorter words from $L_{\textit{pre}}$. Since external contexts only wrap words, and no word in this subset can be derived from another within it, infinitely many $c a^n d$ must be derived from $z = a^m d$ using the context $(c a^{n-m}, \lambda)$ for some $m < n$.
	
	This requires a selection pair $(S, C) \in \cS'$ with $a^m d \in S$ and $(c a^{n-m}, \lambda) \in C$. 
	Since $S \in \SYDEF$, it holds $S = E V_S^* H$ for an alphabet $V_S$ and regular languages $E\subseteq V_S^*$
	and $H\subseteq V_S^*$. Since $a^m d \in S$, the alphabet $V_S$ contains at least the letters $a$ and $d$. 
	The word $a^m d$ factors as $e x h$ with $e \in E$ and $h \in H$. This restricts the factorization to the following 
	three cases:
	
	\textit{Case 1: $h = \lambda$.}
	Then $e x = a^m d$. The letter $d$ is either in $e$ or in $x$.
	\begin{itemize}
		\item \textit{Subcase 1a ($d$ is in $e$):} Then $e = a^m d$ and $x = \lambda$. 
		Thus, $\sets{a^m d} V_S^* \subseteq S$. Since $d \in V_S^*$, the word $w = a^m d d$ belongs to $S$. 
		Since $w \in L_{\textit{pre}}$, it is a word generated by $G'$. Applying the context 
		gives $c a^{n-m} a^m d d = c a^n d d \notin L_{\textit{pre}}$, which is a contradiction.
		\item \textit{Subcase 1b ($d$ is in $x$):} Then $e = a^i$ for some $i \le m$. Thus, $a^i V_S^* \subseteq S$. 
		Since $d a d \in V_S^*$, the word $w = a^i d a d$ belongs to $S$. 
		Since $w \in L_{\textit{pre}}$, it is generated by $G'$. Applying the context 
		gives $c a^{n-m} a^i d a d \notin L_{\textit{pre}}$, which is a contradiction.
	\end{itemize}
	
	\textit{Case 2: $h = d$.}
	Then $e x = a^m$, meaning $e = a^i$ for some $i \le m$. Thus, $a^i V_S^* d \subseteq S$. 
	Since $d a \in V_S^*$, the word $w = a^i (d a) d = a^i d a d$ belongs to $S$. 
	Since it belongs to $L_{\textit{pre}}$, it is generated by $G'$. Applying the 
	context gives $c a^{n-m} a^i d a d \notin L_{\textit{pre}}$, which is a contradiction.
	
	\textit{Case 3: $h = a^j d$ for some $j \ge 1$.}
	Then $e x = a^{m-j}$, meaning $e = a^i$ for some $i \le m-j$. Thus, $a^i V_S^* a^j d \subseteq S$. Since $d \in V_S^*$, the word $w = a^i d a^j d$ belongs to $S$. 
	Since it belongs to $L_{\textit{pre}}$, it is generated by $G'$. Applying the context 
	gives $c a^{n-m} a^i d a^j d \notin L_{\textit{pre}}$, which is a contradiction.
	
	Since all possible factorizations inevitably cause $G'$ to apply the context to a word of $L_{\textit{pre}}$ 
	with two letters $d$, 
	it generates a word starting with $c$ but containing two $d$'s. Since $L_{\textit{pre}}$ restricts words starting 
	with $c$ to have exactly one $d$, a word is generated which does not belong to $L_{\textit{pre}}$. 
	Thus, $L_{\textit{pre}} \notin \ec{\SYDEF}$.
\end{proof}

\begin{lemma}\label{lemma:ec_sydef_closed_under_reversal}
	The language family $\ec{\SYDEF}$ is closed under reversal. Formally, if $L \in \ec{\SYDEF}$, 
	then $L^R \in \ec{\SYDEF}$.
\end{lemma}
\begin{proof}
	Let $L \in \ec{\SYDEF}$. By definition, there exists an external contextual grammar $G = (V, \cS, A)$ such that $L_{\textit{ex}}(G) = L$, and all selection languages in $\cS$ belong to $\SYDEF$. Let $\cS = \{(S_1, C_1), \dots, (S_k, C_k)\}$.
	
	We construct a new external contextual grammar $G^R = (V, \cS^R, A^R)$ to generate the reversed language $L^R$. We define the components of $G^R$ as follows:
	\begin{itemize}
		\item $A^R = \{ w^R \mid w \in A \}$ is the set of reversed axioms.
		\item For each selection pair $(S_i, C_i) \in \cS$, we create a reversed selection pair $(S_i^R, C_i^R) \in \cS^R$, where:
		\begin{itemize}
			\item $S_i^R = \{ w^R \mid w \in S_i \}$,
			\item $C_i^R = \{ (v^R, u^R) \mid (u, v) \in C_i \}$.
		\end{itemize}
	\end{itemize}
	
	First, we must verify that all selection languages in $G^R$ belong to $\SYDEF$. Let $S \in \SYDEF$ be a selection language from $G$. By definition, $S$ can be represented as $S = G_S V_S^* H_S$ for some regular languages $G_S$ and $H_S$ over an alphabet $V_S \subseteq V$. The reversal of $S$ is given by:
	\[ S^R = (G_S V_S^* H_S)^R = H_S^R (V_S^*)^R G_S^R = H_S^R V_S^* G_S^R. \]
	Since the family of regular languages is closed under reversal, both $H_S^R$ and $G_S^R$ are regular languages. 
	Therefore, $S^R$ matches the definition of a symmetric definite language, yielding $S^R \in \SYDEF$.
	
	Next, we show the equality $L_{\textit{ex}}(G^R) = L^R$. We prove this by showing that $w \in L_{\textit{ex}}(G)$ if and only if $w^R \in L_{\textit{ex}}(G^R)$.
	This equivalence holds trivially for the base cases, since $w \in A \iff w^R \in A^R$. 
	For the derivation steps, suppose $w_{j+1}$ is derived from $w_j$ in $G$. This means there is a selection pair $(S_i, C_i) \in \cS$ such that $w_j \in S_i$ and $w_{j+1} = u w_j v$ for some $(u, v) \in C_i$. 
	By our construction of $G^R$, the reversed word $w_j^R$ belongs to $S_i^R$, and the pair $(v^R, u^R)$ belongs to $C_i^R$. Applying this context to $w_j^R$ in $G^R$ yields:
	\[ v^R w_j^R u^R = (u w_j v)^R = w_{j+1}^R. \]
	Thus, every derivation step $w_j \implies w_{j+1}$ in $G$ corresponds exactly to a derivation step $w_j^R \implies w_{j+1}^R$ in $G^R$, and vice versa. 
	
	Consequently, $G^R$ generates exactly all reversed words of $L$. Since $G^R$ is an external contextual grammar 
	with selection in $\SYDEF$, it follows that $L^R = L_{\textit{ex}}(G^R) \in \ec{\SYDEF}$, concluding the proof.
\end{proof}

\begin{lemma}\label{lemma:ec_suff_not_in_sydef}
	Let 
        $L_{\textit{suf}} = L_{\textit{pre}}^R = \set{ d a^n }{ n \ge 0 } \cup \set{ d a^n c }{ n \ge 0 } 
        \cup \set{ d a^m d a^n }{ n, m \ge 0 }$ 
        with $L_{\textit{pre}}$ from Lemma \ref{lemma:ec_pref_not_in_sydef}.
	Then, it holds $L_{\textit{suf}} \in \ec{\SUFF} \setminus \ec{\SYDEF}$.
\end{lemma}
\begin{proof}
	First, we show that $L_{\textit{suf}} \in \ec{\SUFF}$. The external contextual grammar $G = (\{a,c,d\}, \cS, \{d\})$ with the selection pairs $\cS = \{(S_1, C_1), (S_2, C_2)\}$ defined as
	\iirule
	{\set{ d a^n }{ n \ge 0 }}
	{\{(\lambda, a), (\lambda, c), (\lambda, d)\}}
	{\set{ d a^m d a^n }{ n\ge 0,\ m \ge 0 }}
	{\{(\lambda, a)\}}
	generates $L_{\textit{suf}}$. Every word in $S_1$ begins with exactly one $d$. 
	Every word in $S_2$ contains exactly two occurrences of the letter $d$ and begins with the first one. 
	Thus, no word in $S_1$ or $S_2$ can be a proper suffix of another word within the same language, 
	yielding $S_1, S_2 \in \SUFF$. By structural symmetry to Lemma~\ref{lemma:ec_pref_not_in_sydef}, 
	the contexts independently append the respective mirrored characters to the right side, 
	generating the language~$L_{\textit{suf}}$. Thus, $L_{\textit{suf}} \in \ec{\SUFF}$.
	
	Assume that $L_{\textit{suf}} \in \ec{\SYDEF}$ via some grammar $G'$. The property of a language 
	belonging to the class~$\ec{\SYDEF}$ is closed under word reversal because the reverse of any symmetric 
	definite selection language $S = G V^* H$ is $S^R = H^R V^* G^R$, which is again symmetric 
	definite (\cite{Olejar_Szabari.2023}). Since external contextual derivations are completely
	symmetric, with Lemma \ref{lemma:ec_sydef_closed_under_reversal}, it follows that
	$L_{\textit{suf}} \in \ec{\SYDEF}$ implies~$L_{\textit{suf}}^R \in \ec{\SYDEF}$.
	
	Since $L_{\textit{suf}}^R = (L_{\textit{pre}}^R)^R = L_{\textit{pre}}$, this would imply that $L_{\textit{pre}} \in \ec{\SYDEF}$. However, as established in Lemma \ref{lemma:ec_pref_not_in_sydef}, $L_{\textit{pre}} \notin \ec{\SYDEF}$. From this contradiction, it follows that $L_{\textit{suf}} \notin \ec{\SYDEF}$.
\end{proof}

We now prove an equivalence and some inclusion relations.


\begin{lemma}\label{lemma:ec_inff_is_fin}
	The family of languages generated by external contextual grammars with infix-free selection languages coincides with the family of finite languages. Formally, $\ec{\INFF} = \FIN$.
\end{lemma}
\begin{proof}
	Every finite language $L$ over an alphabet $V$ is generated by the external contextual grammar 
	$G = (V, \emptyset, L)$. Since the set of selection rules is empty, the condition is satisfied, yielding $\FIN \subseteq \ec{\INFF}$.
	
	Conversely, let $G = (V, \cS, A)$ be an external contextual grammar with $k$ selection pairs 
	where all selection languages belong to $\INFF$. Assume $L_{ex}(G)$ is infinite. Since the set
	of axioms $A$ and all context sets are finite, the grammar must produce arbitrarily long 
	derivation sequences. In any derivation sequence with more than $k$ steps, the Pigeonhole Principle 
	implies that at least one selection language $S$ is used in two distinct steps $i$ and $j$ (with $i < j$). 
	Thus, the used words $w_i$ and $w_j$ both belong to $S$. Since external contextual derivations increase 
	the word length, $w_i$ is a proper infix of $w_j$. This contradicts $S \in \INFF$. Therefore, the length 
	of any derivation sequence is bounded by $k$. Expanding a finite set of axioms at most $k$ times 
	using finite sets of contexts yields only a finite set of words. Thus, $L_{ex}(G)$ is finite, 
	yielding $\ec{\INFF} \subseteq \FIN$.
\end{proof}

\begin{lemma}\label{lemma:ec_pref_suff_subset_ps}
	The proper inclusions $\ec{\PREF} \subset \ec{\PS}$ and $\ec{\SUFF} \subset \ec{\PS}$ hold.
\end{lemma}
\begin{proof}
	First, we show the inclusions $\ec{\PREF} \subseteq \ec{\PS}$ and $\ec{\SUFF} \subseteq \ec{\PS}$. 
	Since the inclusions $\PREF \subseteq \PS$ and $\SUFF \subseteq \PS$ hold for the families of selection languages, any external contextual grammar with selection languages in $\PREF$ or $\SUFF$ is also an external contextual grammar with selection languages in $\PS$. Thus, every language generated by the former can also be generated by the latter, yielding $\ec{\PREF} \subseteq \ec{\PS}$ and $\ec{\SUFF} \subseteq \ec{\PS}$.
	
	To show that these inclusions are proper, we consider the language $L = \{a\}^*$ from Lemma \ref{lemma:ec_mon_setminus_pref_suff}. 
	As established in that lemma, $L \in \ec{\MON}$. Since $\MON \subseteq \PS$, it follows that $\ec{\MON} \subseteq \ec{\PS}$, which yields $L \in \ec{\PS}$. 
	Furthermore, Lemma \ref{lemma:ec_mon_setminus_pref_suff} proves that $L \notin \ec{\PREF} \cup \ec{\SUFF}$. 
	
	Consequently, $L \in \ec{\PS} \setminus \ec{\PREF}$ and $L \in \ec{\PS} \setminus \ec{\SUFF}$, proving that the inclusions are proper.
\end{proof}

\begin{corollary}\label{cor:ec_inclusions}
	The proper inclusions $\ec{\INFF} \subset \ec{\PREF}$ and $\ec{\INFF} \subset \ec{\SUFF}$ hold.
\end{corollary}
\begin{proof}
	By Lemma \ref{lemma:ec_inff_is_fin}, it holds $\ec{\INFF} = \FIN$. Every finite language is generated by an external contextual grammar with an empty set of selection rules, satisfying the requirement for prefix-free or suffix-free selection languages. Thus, $\ec{\INFF} \subseteq \ec{\PREF}$ and $\ec{\INFF} \subseteq \ec{\SUFF}$ by Lemma \ref{lemma:ec_monoton}. 
	
	Since $\ec{\PREF}$ and $\ec{\SUFF}$ contain infinite languages (as shown by the witness languages, e.g. in Lemma \ref{lemma:ec_pref_suff_setminus_pre_suf_inf}), while $\ec{\INFF}$ contains only finite languages, both inclusions are proper.
\end{proof}

\begin{corollary}\label{cor:ec_inff_eq_fin}
	The language family $\ec{\INFF}$ coincides with the family $\ec{\FIN}$.
\end{corollary}
\begin{proof}
	By Lemma \ref{lemma:ec_inff_is_fin}, it holds $\ec{\INFF} = \FIN$. Since $\FIN = \ec{\FIN}$ 
	(as shown in \cite{Istrail.1978}), the equality $\ec{\INFF} = \ec{\FIN}$ follows directly.
\end{proof}

With the languages from the previous lemmas, the incomparabilities depicted in Figure~\ref{fig:ec_erg} can be shown.

\begin{theorem}[Resulting hierarchy for $\cEC$]\label{theorem:neue_hierarchie_EC}
The inclusion relations presented in Figure \ref{fig:ec_erg} hold. An arrow from an entry $X$ to
an entry~$Y$ depicts the proper inclusion $X \subset Y$; if two families are not connected by a directed
path, they are incomparable.
\end{theorem}

\section{Conclusion and future work}

In this paper, we have continued the investigation of subregular language families and their application 
as selection languages in contextual grammars. We have introduced the families of prefix-free, suffix-free, 
and infix-free languages (idefix-free languages) and established their exact positions within the existing 
hierarchy of subregular language families. Furthermore, we have comprehensively examined the generative 
capacity of external contextual grammars regulated by these new families, extending the known inclusion diagrams. 

For future work, it remains an ongoing effort to complete the hierarchies of subregular language families and the corresponding families of externally and internally generated contextual languages. As previously noted, the extension of the hierarchy with various other families of definite-like languages (for instance, ultimate definite, central definite, and non-initial definite languages) is currently under investigation. Furthermore, it is planned to unify the extended hierarchy of subregular language families with the hierarchies of language families generated by contextual grammars defined by limited resources (e.\,g., the number of contexts, the number of selection rules, or the size of the contexts). Finally, applying these specific subregular restrictions to other formal frameworks, such as tree-controlled grammars or networks of evolutionary processors, presents an intriguing direction for upcoming studies. In addition, answering various decidability and complexity questions regarding the families investigated in this work is high on our research agenda.


\begin{thebibliography}{10}
\providecommand{\bibitemdeclare}[2]{}
\providecommand{\surnamestart}{}
\providecommand{\surnameend}{}
\providecommand{\urlprefix}{Available at }
\providecommand{\url}[1]{\texttt{#1}}
\providecommand{\href}[2]{\texttt{#2}}
\providecommand{\urlalt}[2]{\href{#1}{#2}}
\providecommand{\doi}[1]{doi:\urlalt{https://doi.org/#1}{#1}}
\providecommand{\eprint}[1]{arXiv:\urlalt{https://arxiv.org/abs/#1}{#1}}
\providecommand{\bibinfo}[2]{#2}

\bibitemdeclare{article}{Bordihn_Holzer_Kutrib.2009}
\bibitem{Bordihn_Holzer_Kutrib.2009}
\bibinfo{author}{Henning \surnamestart Bordihn\surnameend},
  \bibinfo{author}{Markus \surnamestart Holzer\surnameend} \&
  \bibinfo{author}{Martin \surnamestart Kutrib\surnameend}
  (\bibinfo{year}{2009}): \emph{\bibinfo{title}{Determination of finite
  automata accepting subregular languages}}.
\newblock {\slshape \bibinfo{journal}{Theoretical Computer Science}}
  \bibinfo{volume}{410}(\bibinfo{number}{35}), pp. \bibinfo{pages}{3209--3222},
  \doi{10.1016/j.tcs.2009.05.019}.

\bibitemdeclare{phdthesis}{Brzozowski.1962}
\bibitem{Brzozowski.1962}
\bibinfo{author}{Janusz~A. \surnamestart Brzozowski\surnameend}
  (\bibinfo{year}{1962}): \emph{\bibinfo{title}{Regular expression techniques
  for sequential circuits}}.
\newblock Ph.D. thesis, \bibinfo{school}{Princeton University, Princeton, NJ,
  USA}.

\bibitemdeclare{article}{Brzozowski.1967}
\bibitem{Brzozowski.1967}
\bibinfo{author}{Janusz~A. \surnamestart Brzozowski\surnameend}
  (\bibinfo{year}{1967}): \emph{\bibinfo{title}{Roots of star events}}.
\newblock {\slshape \bibinfo{journal}{Journal of the ACM}}
  \bibinfo{volume}{14}(\bibinfo{number}{3}), pp. \bibinfo{pages}{466--477},
  \doi{10.1109/SWAT.1966.21}.

\bibitemdeclare{article}{Brzozowski_Cohen.1969}
\bibitem{Brzozowski_Cohen.1969}
\bibinfo{author}{Janusz~A. \surnamestart Brzozowski\surnameend} \&
  \bibinfo{author}{Rina \surnamestart Cohen\surnameend} (\bibinfo{year}{1969}):
  \emph{\bibinfo{title}{On decompositions of regular events}}.
\newblock {\slshape \bibinfo{journal}{Journal of the ACM}}
  \bibinfo{volume}{16}(\bibinfo{number}{1}), pp. \bibinfo{pages}{132--144},
  \doi{10.1145/321495.321505}.

\bibitemdeclare{article}{Brzozowski_Jiraskova_Zou.2014}
\bibitem{Brzozowski_Jiraskova_Zou.2014}
\bibinfo{author}{Janusz~A. \surnamestart Brzozowski\surnameend},
  \bibinfo{author}{Galina \surnamestart Jir\'askov\'a\surnameend} \&
  \bibinfo{author}{Chenglong \surnamestart Zou\surnameend}
  (\bibinfo{year}{2014}): \emph{\bibinfo{title}{Quotient complexity of closed
  languages}}.
\newblock {\slshape \bibinfo{journal}{Theory of Computing Systems}}
  \bibinfo{volume}{54}, pp. \bibinfo{pages}{277--292},
  \doi{10.1007/s00224-013-9515-7}.

\bibitemdeclare{article}{Dassow.2005}
\bibitem{Dassow.2005}
\bibinfo{author}{J{\"{u}}rgen \surnamestart Dassow\surnameend}
  (\bibinfo{year}{2005}): \emph{\bibinfo{title}{Contextual grammars with
  subregular choice}}.
\newblock {\slshape \bibinfo{journal}{Fundamenta Informaticae}}
  \bibinfo{volume}{64}(\bibinfo{number}{1--4}), pp. \bibinfo{pages}{109--118}.

\bibitemdeclare{article}{Dassow.2015}
\bibitem{Dassow.2015}
\bibinfo{author}{J{\"u}rgen \surnamestart Dassow\surnameend}
  (\bibinfo{year}{2015}): \emph{\bibinfo{title}{Contextual languages with
  strictly locally testable and star free selection languages}}.
\newblock {\slshape \bibinfo{journal}{Analele Universitatii Bucuresti}}
  \bibinfo{volume}{62}, pp. \bibinfo{pages}{25--36}.

\bibitemdeclare{article}{Dassow_Manea_Truthe.2012}
\bibitem{Dassow_Manea_Truthe.2012}
\bibinfo{author}{J{\"{u}}rgen \surnamestart Dassow\surnameend},
  \bibinfo{author}{Florin \surnamestart Manea\surnameend} \&
  \bibinfo{author}{Bianca \surnamestart Truthe\surnameend}
  (\bibinfo{year}{2012}): \emph{\bibinfo{title}{On external contextual grammars
  with subregular selection languages}}.
\newblock {\slshape \bibinfo{journal}{Theoretical Computer Science}}
  \bibinfo{volume}{449}, pp. \bibinfo{pages}{64--73},
  \doi{10.1016/j.tcs.2012.04.008}.

\bibitemdeclare{article}{DasManTru12b}
\bibitem{DasManTru12b}
\bibinfo{author}{J{\"u}rgen \surnamestart Dassow\surnameend},
  \bibinfo{author}{Florin \surnamestart Manea\surnameend} \&
  \bibinfo{author}{Bianca \surnamestart Truthe\surnameend}
  (\bibinfo{year}{2012}): \emph{\bibinfo{title}{On Subregular Selection
  Languages in Internal Contextual Grammars}}.
\newblock {\slshape \bibinfo{journal}{Journal of Automata, Languages, and
  Combinatorics}} \bibinfo{volume}{17}(\bibinfo{number}{2--4}), pp.
  \bibinfo{pages}{145--164}, \doi{10.25596/jalc-2012-145}.

\bibitemdeclare{article}{Dassow_Truthe.2023}
\bibitem{Dassow_Truthe.2023}
\bibinfo{author}{J{\"{u}}rgen \surnamestart Dassow\surnameend} \&
  \bibinfo{author}{Bianca \surnamestart Truthe\surnameend}
  (\bibinfo{year}{2023}): \emph{\bibinfo{title}{Relations of contextual
  grammars with strictly locally testable selection languages}}.
\newblock {\slshape \bibinfo{journal}{{RAIRO} -- Theoretical Informatics and
  Applications}} \bibinfo{volume}{57}, p. \bibinfo{pages}{\#10},
  \doi{10.1051/ita/2023012}.

\bibitemdeclare{book}{Gecseg_Peak.1972}
\bibitem{Gecseg_Peak.1972}
\bibinfo{author}{Ference \surnamestart G\'ecseg\surnameend} \&
  \bibinfo{author}{Istv{\'{a}}n \surnamestart Pe{\'{a}}k\surnameend}
  (\bibinfo{year}{1972}): \emph{\bibinfo{title}{Algebraic Theory of Automata}}.
\newblock \bibinfo{publisher}{Academiai Kiado, Budapest}.

\bibitemdeclare{article}{Gill_Kou.1974}
\bibitem{Gill_Kou.1974}
\bibinfo{author}{Arthur \surnamestart Gill\surnameend} \&
  \bibinfo{author}{Lawrence~T. \surnamestart Kou\surnameend}
  (\bibinfo{year}{1974}): \emph{\bibinfo{title}{Multiple-entry finite
  automata}}.
\newblock {\slshape \bibinfo{journal}{Journal of Computer and System Sciences}}
  \bibinfo{volume}{9}(\bibinfo{number}{1}), pp. \bibinfo{pages}{1--19},
  \doi{10.1016/S0022-0000(74)80034-6}.

\bibitemdeclare{article}{Havel.1969}
\bibitem{Havel.1969}
\bibinfo{author}{Ivan~M. \surnamestart Havel\surnameend}
  (\bibinfo{year}{1969}): \emph{\bibinfo{title}{The theory of regular events
  {II}}}.
\newblock {\slshape \bibinfo{journal}{Kybernetika}}
  \bibinfo{volume}{5}(\bibinfo{number}{6}), pp. \bibinfo{pages}{520--544}.

\bibitemdeclare{inproceedings}{Holzer_Truthe.2015}
\bibitem{Holzer_Truthe.2015}
\bibinfo{author}{Markus \surnamestart Holzer\surnameend} \&
  \bibinfo{author}{Bianca \surnamestart Truthe\surnameend}
  (\bibinfo{year}{2015}): \emph{\bibinfo{title}{On relations between some
  subregular language families}}.
\newblock In \bibinfo{editor}{Rudolf \surnamestart Freund\surnameend},
  \bibinfo{editor}{Markus \surnamestart Holzer\surnameend},
  \bibinfo{editor}{Nelma \surnamestart Moreira\surnameend} \&
  \bibinfo{editor}{Rog{\'{e}}rio \surnamestart Reis\surnameend}, editors:
  {\slshape \bibinfo{booktitle}{Seventh Workshop on Non-Classical Models of
  Automata and Applications -- \hbox{NCMA} 2015, Porto, Portugal, August 31 --
  September 1, 2015. Proceedings}}, {\slshape \bibinfo{series}{books@ocg.at}}
  \bibinfo{volume}{318}, \bibinfo{publisher}{{\"{O}}sterreichische Computer
  Gesellschaft}, pp. \bibinfo{pages}{109--124}.

\bibitemdeclare{article}{Istrail.1978}
\bibitem{Istrail.1978}
\bibinfo{author}{Sorin \surnamestart Istrail\surnameend}
  (\bibinfo{year}{1978}): \emph{\bibinfo{title}{Gramatici contextuale cu
  selectiva regulata}}.
\newblock {\slshape \bibinfo{journal}{Stud. Cerc. Mat}} \bibinfo{volume}{30},
  pp. \bibinfo{pages}{287--294}.

\bibitemdeclare{inproceedings}{Koedding.Truthe.2024}
\bibitem{Koedding.Truthe.2024}
\bibinfo{author}{Marvin \surnamestart K{\"{o}}dding\surnameend} \&
  \bibinfo{author}{Bianca \surnamestart Truthe\surnameend}
  (\bibinfo{year}{2024}): \emph{\bibinfo{title}{Various Types of Comet
  Languages and their Application in External Contextual Grammars}}.
\newblock In \bibinfo{editor}{Florin \surnamestart Manea\surnameend} \&
  \bibinfo{editor}{Giovanni \surnamestart Pighizzini\surnameend}, editors:
  {\slshape \bibinfo{booktitle}{Proceedings 14th International Workshop on
  Non-Classical Models of Automata and Applications {(NCMA} 2024), {NCMA} 2024,
  G{\"{o}}ttingen, Germany, 12--13 August 2024}}, {\slshape
  \bibinfo{series}{{EPTCS}}} \bibinfo{volume}{407}, pp.
  \bibinfo{pages}{118--135}, \doi{10.4204/EPTCS.407.9}.

\bibitemdeclare{article}{Koedding.Truthe.2025.idefix}
\bibitem{Koedding.Truthe.2025.idefix}
\bibinfo{author}{Marvin \surnamestart K{\"{o}}dding\surnameend} \&
  \bibinfo{author}{Bianca \surnamestart Truthe\surnameend}
  (\bibinfo{year}{\noop{3001}submitted}): \emph{\bibinfo{title}{Idefix-Closed
  Languages and Their Application in Contextual Grammars}}.
\newblock {\slshape \bibinfo{journal}{{RAIRO} -- Theoretical Informatics and
  Applications}}.

\bibitemdeclare{article}{Koedding.Truthe.NCMA24.JALC.2025}
\bibitem{Koedding.Truthe.NCMA24.JALC.2025}
\bibinfo{author}{Marvin \surnamestart K{\"{o}}dding\surnameend} \&
  \bibinfo{author}{Bianca \surnamestart Truthe\surnameend}
  (\bibinfo{year}{\noop{3001}submitted}): \emph{\bibinfo{title}{Various Types
  of Comet Languages and Their Application in Contextual Grammars}}.
\newblock {\slshape \bibinfo{journal}{Journal of Automata, Languages, and
  Combinatorics}}.

\bibitemdeclare{incollection}{Kudlek.2004}
\bibitem{Kudlek.2004}
\bibinfo{author}{Manfred \surnamestart Kudlek\surnameend}
  (\bibinfo{year}{2004}): \emph{\bibinfo{title}{On languages of cyclic words}}.
\newblock In \bibinfo{editor}{Natasha \surnamestart Jonoska\surnameend},
  \bibinfo{editor}{{\relax Gh}eorghe \surnamestart P\u{a}un\surnameend} \&
  \bibinfo{editor}{Grzegorz \surnamestart Rozenberg\surnameend}, editors:
  {\slshape \bibinfo{booktitle}{Aspects of Molecular Computing, Essays
  Dedicated to Tom Head on the Occasion of His 70th Birthday}}, {\slshape
  \bibinfo{series}{LNCS}} \bibinfo{volume}{2950},
  \bibinfo{publisher}{Springer-Verlag}, pp. \bibinfo{pages}{278--288},
  \doi{10.1007/978-3-540-24635-0_20}.

\bibitemdeclare{article}{Marcus.1969}
\bibitem{Marcus.1969}
\bibinfo{author}{Solomon \surnamestart Marcus\surnameend}
  (\bibinfo{year}{1969}): \emph{\bibinfo{title}{Contextual grammars}}.
\newblock {\slshape \bibinfo{journal}{Revue Roumaine de Math{\'{e}}matique
  Pures et Appliqu{\'{e}}es}} \bibinfo{volume}{14}, pp.
  \bibinfo{pages}{1525--1534}.

\bibitemdeclare{book}{McNaughton_Papert.1971}
\bibitem{McNaughton_Papert.1971}
\bibinfo{author}{Robert \surnamestart McNaughton\surnameend} \&
  \bibinfo{author}{Seymour \surnamestart Papert\surnameend}
  (\bibinfo{year}{1971}): \emph{\bibinfo{title}{Counter-Free Automata}}.
\newblock \bibinfo{publisher}{MIT Press}, \bibinfo{address}{Cambridge, USA}.

\bibitemdeclare{inproceedings}{Nagy.2019}
\bibitem{Nagy.2019}
\bibinfo{author}{Benedek \surnamestart Nagy\surnameend} (\bibinfo{year}{2019}):
  \emph{\bibinfo{title}{Union-freeness, deterministic union-freeness and
  union-complexity}}.
\newblock In \bibinfo{editor}{Michal \surnamestart Hospod{\'a}r\surnameend},
  \bibinfo{editor}{Galina \surnamestart Jir{\'a}skov{\'a}\surnameend} \&
  \bibinfo{editor}{Stavros \surnamestart Konstantinidis\surnameend}, editors:
  {\slshape \bibinfo{booktitle}{Descriptional Complexity of Formal Systems,
  21st IFIP WG 1.02 International Conference, DCFS 2019, Ko{\v{s}}ice,
  Slovakia, July 17--19, 2019, Proceedings}}, \bibinfo{publisher}{Springer,
  Cham}, pp. \bibinfo{pages}{46--56}, \doi{10.1007/978-3-030-23247-4_3}.

\bibitemdeclare{article}{Olejar_Szabari.2023}
\bibitem{Olejar_Szabari.2023}
\bibinfo{author}{Viktor \surnamestart Olej{\'{a}}r\surnameend} \&
  \bibinfo{author}{Alexander \surnamestart Szabari\surnameend}
  (\bibinfo{year}{2025}): \emph{\bibinfo{title}{Closure Properties of
  Subregular Languages Under Operations}}.
\newblock {\slshape \bibinfo{journal}{Int. J. Found. Comput. Sci.}}
  \bibinfo{volume}{36}(\bibinfo{number}{7}), pp. \bibinfo{pages}{1063--1087},
  \doi{10.1142/S0129054123450016}.

\bibitemdeclare{article}{Paz_Peleg.1965}
\bibitem{Paz_Peleg.1965}
\bibinfo{author}{Azaria \surnamestart Paz\surnameend} \&
  \bibinfo{author}{Bezalel \surnamestart Peleg\surnameend}
  (\bibinfo{year}{1965}): \emph{\bibinfo{title}{Ultimate-definite and
  symmetric-definite events and automata}}.
\newblock {\slshape \bibinfo{journal}{Journal of the ACM}}
  \bibinfo{volume}{12}(\bibinfo{number}{3}), pp. \bibinfo{pages}{399--410},
  \doi{10.1145/321281.321292}.

\bibitemdeclare{article}{Perles_Rabin_Shamir.1963}
\bibitem{Perles_Rabin_Shamir.1963}
\bibinfo{author}{Micha~A. \surnamestart Perles\surnameend},
  \bibinfo{author}{Michael~O. \surnamestart Rabin\surnameend} \&
  \bibinfo{author}{Eli \surnamestart Shamir\surnameend} (\bibinfo{year}{1963}):
  \emph{\bibinfo{title}{The theory of definite automata}}.
\newblock {\slshape \bibinfo{journal}{IEEE Transactions of Electronic
  Computers}} \bibinfo{volume}{12}, pp. \bibinfo{pages}{233--243},
  \doi{10.1109/PGEC.1963.263534}.

\bibitemdeclare{book}{Rozenberg_Salomaa.1997}
\bibitem{Rozenberg_Salomaa.1997}
\bibinfo{editor}{Grzegorz \surnamestart Rozenberg\surnameend} \&
  \bibinfo{editor}{Arto \surnamestart Salomaa\surnameend}, editors
  (\bibinfo{year}{1997}): \emph{\bibinfo{title}{Handbook of Formal Languages}}.
\newblock \bibinfo{publisher}{Springer-Verlag}, \bibinfo{address}{Berlin},
  \doi{10.1007/978-3-642-59136-5}.

\bibitemdeclare{book}{Shyr.1991}
\bibitem{Shyr.1991}
\bibinfo{author}{Huei{-}Jan \surnamestart Shyr\surnameend}
  (\bibinfo{year}{1991}): \emph{\bibinfo{title}{Free Monoids and Languages}}.
\newblock \bibinfo{publisher}{Hon Min Book Co., Taichung, Taiwan}.

\bibitemdeclare{article}{Shyr_Thierrin.1974.ord}
\bibitem{Shyr_Thierrin.1974.ord}
\bibinfo{author}{Huei{-}Jan \surnamestart Shyr\surnameend} \&
  \bibinfo{author}{Gabriel \surnamestart Thierrin\surnameend}
  (\bibinfo{year}{1974}): \emph{\bibinfo{title}{Ordered automata and associated
  languages}}.
\newblock {\slshape \bibinfo{journal}{Tamkang Journal of Mathematics}}
  \bibinfo{volume}{5}(\bibinfo{number}{1}), pp. \bibinfo{pages}{9--20}.

\bibitemdeclare{article}{Shyr_Thierrin.1974.ps}
\bibitem{Shyr_Thierrin.1974.ps}
\bibinfo{author}{Huei{-}Jan \surnamestart Shyr\surnameend} \&
  \bibinfo{author}{Gabriel \surnamestart Thierrin\surnameend}
  (\bibinfo{year}{1974}): \emph{\bibinfo{title}{Power-separating regular
  languages}}.
\newblock {\slshape \bibinfo{journal}{Mathematical Systems Theory}}
  \bibinfo{volume}{8}(\bibinfo{number}{1}), pp. \bibinfo{pages}{90--95},
  \doi{10.1007/BF01761710}.

\bibitemdeclare{inproceedings}{Truthe.2014}
\bibitem{Truthe.2014}
\bibinfo{author}{Bianca \surnamestart Truthe\surnameend}
  (\bibinfo{year}{2014}): \emph{\bibinfo{title}{A relation between definite and
  ordered finite automata}}.
\newblock In \bibinfo{editor}{Suna \surnamestart Bensch\surnameend},
  \bibinfo{editor}{Rudolf \surnamestart Freund\surnameend} \&
  \bibinfo{editor}{Friedrich \surnamestart Otto\surnameend}, editors: {\slshape
  \bibinfo{booktitle}{Sixth Workshop on Non-Classical Models for Automata and
  Applications -- {NCMA} 2014, Kassel, Germany, July 28--29, 2014.
  Proceedings}}, {\slshape \bibinfo{series}{books@ocg.at}}
  \bibinfo{volume}{304}, \bibinfo{publisher}{{\"{O}}sterreichische Computer
  Gesellschaft}, pp. \bibinfo{pages}{235--247}.

\bibitemdeclare{techreport}{Truthe.2018}
\bibitem{Truthe.2018}
\bibinfo{author}{Bianca \surnamestart Truthe\surnameend}
  (\bibinfo{year}{2018}): \emph{\bibinfo{title}{Hierarchy of Subregular
  Language Families}}.
\newblock \bibinfo{type}{Technical Report},
  \bibinfo{institution}{Justus-Liebig-Universit{\"{a}}t Giessen, Institut
  f{\"{u}}r Informatik, IFIG Research Report 1801}.

\bibitemdeclare{article}{Truthe.2021}
\bibitem{Truthe.2021}
\bibinfo{author}{Bianca \surnamestart Truthe\surnameend}
  (\bibinfo{year}{2021}): \emph{\bibinfo{title}{Generative capacity of
  contextual grammars with subregular selection languages}}.
\newblock {\slshape \bibinfo{journal}{Fundamenta Informaticae}}
  \bibinfo{volume}{180}(\bibinfo{number}{1--2}), pp. \bibinfo{pages}{123--150},
  \doi{10.3233/FI-2021-2037}.

\bibitemdeclare{book}{Wiedemann.1978}
\bibitem{Wiedemann.1978}
\bibinfo{author}{Barbara \surnamestart Wiedemann\surnameend}
  (\bibinfo{year}{1978}): \emph{\bibinfo{title}{Vergleich der
  {L}eistungsf{\"a}higkeit endlicher determinierter {A}utomaten}}.
\newblock \bibinfo{publisher}{Diplomarbeit, Universit{\"a}t Rostock}.

\end{thebibliography}
 \newcommand{\noop}[1]{}

\end{document}